\documentclass[journal]{IEEEtran}
\usepackage{amsmath}
\usepackage{amsthm}
\usepackage{subcaption}
\newtheorem{theorem}{Theorem}
\newtheorem{proposition}{Proposition}
\newtheorem{remark}{Remark}
\newtheorem{lemma}{Lemma}
\usepackage{multirow}
\usepackage{booktabs} 
\usepackage{graphicx}
\usepackage{longtable}
\usepackage{xcolor}
\usepackage[T1]{fontenc}
\usepackage{times}
\usepackage{float}
\usepackage{algpseudocode}

\usepackage[ruled]{algorithm} 
\usepackage{amssymb}
\usepackage[margin=0.5in]{geometry}
\usepackage{array}
\usepackage{amssymb}
\begin{document}

\title{FLoRa: Flight-Assisted Data Collection from Duty Cycling LoRa Nodes under Energy Constraints}

\author{Naresh~Babu~Kakarla and 
        V.~Mahendran
\thanks{Naresh Babu Kakarla and V.~Mahendran are with Department
of Computer Science and Engineering, Indian Institute of Technology Tirupati,
Andhra Pradesh, 517619,~India.~(e-mail: \{cs19d501, mahendran\}@iittp.ac.in).}
}

\markboth{Technical Report}%
{Kakarla \MakeLowercase{\textit{et al.}}: FLoRa: Flight-Assisted Data Collection from Duty Cycling LoRa Nodes}

\maketitle
\begin{abstract}
Data collection using Unmanned Aerial Vehicles (UAVs) is challenging when LoRa IoT Devices (IoTDs) duty-cycle to conserve battery. Under energy constraints, a UAV must decide which IoTDs to visit, in what order, where to hover, and how many times to probe each node, while time-based data freshness decays. Tractably solving this problem requires a multi-level optimization architecture: discrete combinatorial optimization for routing, continuous global optimization for spatial positioning, and sequential decision-making under uncertainty. We propose FLoRa, a Flight-assisted LoRa data collection architecture using Simulated Annealing (SA) for path planning, Covariance Matrix Adaptation Evolution Strategy (CMA-ES) for hover positioning, and Partially Observable Markov Decision Processes (POMDPs) for probing IoTDs. To quantify collection utility from duty-cycling nodes, we introduce the Value of Information for Pull-based systems (VIP), a metric that rewards fresh data and penalizes failed probes, imposing well-posedness and preventing indefinite probing when an IoTD is off. Tracking hard battery constraints on every POMDP sample path requires state augmentation, worsening the curse of dimensionality. For tractability, SA and CMA-ES work on the hard battery constraints, while at the POMDP layer we relax them into soft average constraints via Lagrangian relaxation. Since solving the network-wide POMDP is computationally complex, we decompose it into node-level POMDPs by approximating inter-node time dependency using a forward-decomposition technique. Evaluation shows that FLoRa outperforms greedy, metaheuristic, and deep reinforcement learning baselines, as well as recent UAV data collection methods, by $27.8$--$76.1\%$ in total expected VIP, $23.5$--$45.7\%$ in node coverage, and $23.3$--$54.4\%$ in successful collections.
\end{abstract}


\IEEEpeerreviewmaketitle

\section{Introduction}
Collecting data in difficult-to-reach regions with Low Power Wide Area Networks (LPWANs) is essential for applications such as environmental monitoring and disaster management~\cite{IoT}. Battery-powered IoTDs in LPWANs typically operate on a duty cycle that alternates between active and sleep states to conserve battery and extend network lifetime~\cite{LPWAN_DC_SNs}. This duty-cycling behavior, which directly influences the data availability, plays a vital role in the design of data-collection algorithms.

\par Data collection in challenged networks such as LPWANs requires careful planning of both spatial (IoTD visits) and temporal (duty-cycling patterns)~\cite{TC_UAV_DC_5} aspects. For long-range IoT applications, LoRa is considered a viable option because it provides a range of up to $15$ km while maintaining ultra-low power consumption~\cite{LoRa_for_IoT}. Unmanned Aerial Vehicle (UAV)-assisted data collection is especially advantageous when fixed infrastructure is not viable~\cite{UAV_AoI_1}. Moreover, UAVs can traverse large areas, establish Line-of-Sight (LoS) communication with ground IoTDs, and dynamically adjust trajectories based on network conditions.

\par In UAV-assisted LoRa networks, timely data collection is critical. For instance, in disaster response scenarios, the value of collected data typically decreases over time~\cite{Time_varying_data_value}. To measure the time-based utility of data and incorporate it into optimization models, several freshness metrics have been introduced in the context of UAV communications~\cite{UAV_AoI_3,UAV_AoI_4,UoI,VoI_with_deterministic_deadline}. Among them, deadline-based metrics such as Urgency of Information (UoI)~\cite{UoI} and Value of Information (VoI)~\cite{VoI_with_deterministic_deadline} capture data-collection utility that is suitable for disaster-recovery scenarios. These metrics decay over time and become stale after a specified deadline.

\par Existing deadline-based freshness metrics (such as VoI), however, have the following modeling limitation: VoI is not well-defined when data collection fails (for instance, due to device unavailability), thereby making the problem ill-posed. In pull-based lossy-networked systems, the UAV actively probes IoTDs that may not be available due to their duty-cycling on/off patterns. In such scenarios, path-planning algorithms need a negative feedback to account for failed probes. To address this gap, we propose a novel value-based freshness metric namely, the Value of Information for Pull-based systems (VIP), which behaves like VoI upon successful collection, and additionally penalizes failed probe attempts thereby making the problem \textit{well-posed} especially when the IoTD is unavailable by preventing the UAV from probing indefinitely.
\begin{table}[t]
\centering
\caption{ Types of UAV path-planning problems}
\label{tab:detailed_comparison}
\begin{tabular}{|c|l|l|l|l|l|}
\hline
\textbf{Ref}&\textbf{Problem} & \textbf{Path} &\textbf{Positioning}  &\textbf{Type} &\textbf{Probing}  \\
\hline
\cite{TSP_2}&TSP & All & Above IoTDs & Push & Deterministic \\
\hline
\cite{CE_TSP}&CE-TSP & All & IoTD neighborhood &  Push &  Deterministic\\ 
\hline
\cite{OP_DC}&OP & Subset & Above IoTDs &  Push & Deterministic \\
\hline
\cite{CEOP}&CEOP & Subset & IoTD  neighborhood &  Push & Deterministic\\
\hline
\textbf{FLoRa}&\textbf{PTCOP}&\textbf{Subset}&\textbf{IoTD neighborhood}& \textbf{Pull} &\textbf{Probabilistic}\\\hline
\end{tabular}
\end{table}

\par Table~\ref{tab:detailed_comparison} positions the proposed work among the other UAV data collection  problems from the literature. Prior works on UAV path-planning, such as~\cite{CE_TSP,TSP_1}, have implicitly assumed that UAVs visit all of the ground IoTDs. However, the fundamental challenge in UAV-assisted networks is to perform data collection under limited hard battery constraints~\cite{Avg_Energy_const_TITS},~\cite{Avg_Energy_const_MDPI}. For geographically widespread networks (such as LPWANs), such an energy-constrained path-planning algorithm can restrict UAVs to visiting only a subset of IoTDs. The aforementioned subset selection problem is called the Orienteering Problem (OP)~\cite{Time_varying_data_value, Orienteering_Problem, UAV_OP_1, UAV_OP_2} which is NP-hard.

\par As UAVs fly over a large coverage area spanned by the ground IoTDs, the Close-Enough Orienteering Problem (CEOP)~\cite{CEOP} characterizes the nature of data collection by hovering from a distance within the coverage neighborhood of the IoTDs (rather than directly hovering over the IoTDs).  While such an approach saves travel costs and enables collecting more data within an energy-limited travel budget, finding a suitable hovering point within a continuum neighborhood region is computationally intractable.

\par Problems dealing with the timely collection of data coupled with UAV path-planning to visit a subset of IoTDs are classified under the Time-Dependent Orienteering Problem (TDOP)~\cite{TDOP} framework. Moreover, the presence of  duty-cycling IoTDs such as LoRa devices induces uncertainty in the availability of IoTDs. The aforementioned problem frameworks are classified under Probabilistic Orienteering Problems (POP)~\cite{POP}. 
\par To the best of our knowledge, our proposed Flight-assisted data collection over LoRa nodes (FLoRa) is the first hierarchical optimization framework to jointly solve the following: (i) maximizing well-posed time-dependent utility via the proposed VIP freshness metric suitable for duty-cycling IoTDs, (ii) deterministic path-planning and close-enough positioning (hovering),  and (iii) probabilistic probing in the presence of duty-cycling IoTDs. For the first time, the proposed FLoRa framework uniquely solves the UAV data collection problem under the Probabilistic, Time-dependent, and Close-enough Orienteering Problem (PTCOP) setting.
\begin{figure}[t]
        \centering
\includegraphics[scale=.3]{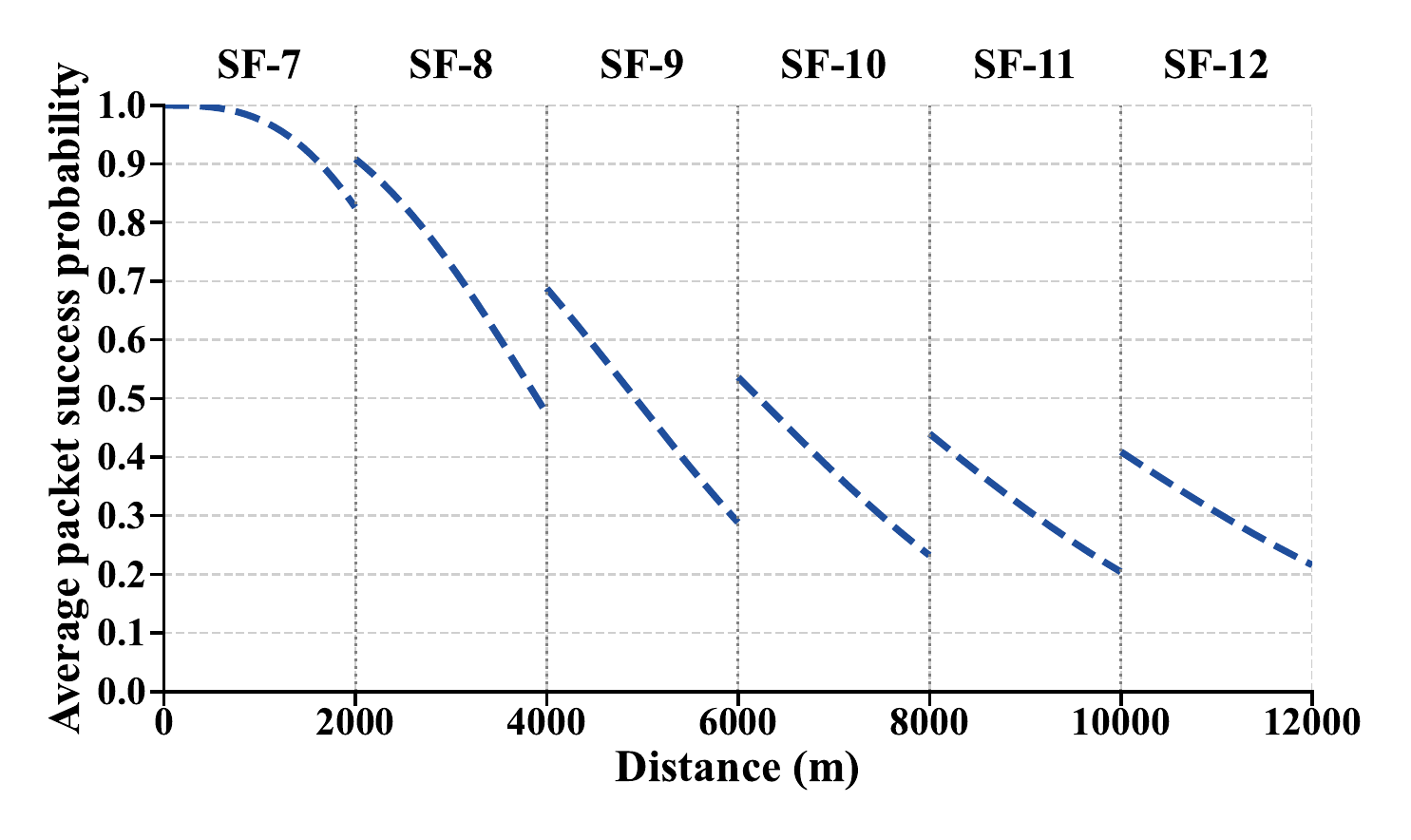}
    \caption{Average packet success probability versus distance for the different SF regions.} 
    \label{fig:snrsuccprb}
\end{figure}

\subsection{Challenges addressed by  FLoRa framework} The PTCOP problem nature of FLoRa needs to deal with various levels of computational challenges as follows:
\par Firstly, the IoTD subset selection problem of determining IoTDs to visit under limited battery constraints of the UAV is a combinatorial optimization problem which is typically NP-hard. Over the selected subset of IoTDs, optimizing UAV hover positions in each of the continuous neighborhood of LoRa IoTDs is NP-hard due to the continuous search space. In addition, the uncertain nature of the response to the UAV's probes introduces a temporal dilemma of ``whether to probe same IoTD again, or move on to a different IoTD?''. Determining the optimal number of probing attempts at each hovering location under random IoTD availability is a challenging sequential decision-making problem. 

\par Secondly, the Signal-to-Noise Ratio (SNR) of IoTDs equipped with LoRa radios is characterized by a step-like function due to the discrete-type Spreading Factor (SF). As a consequence, as shown in Fig.~\ref{fig:snrsuccprb}, the average packet success probability (as a function of distance) is discontinuous at the discrete jumps of the SF transitions. This non-differentiable nature prohibits the use of gradient-based optimization frameworks.
\par Lastly, coverage, freshness, and data collection probability compete with each other. For instance, flying closer to IoTDs increases packet success probability, but this consumes more flight time, thereby reducing freshness. In a similar way, visiting more IoTDs by probing at their coverage boundaries may negatively affect the data collection probability. The VIP objective captures these competing effects in a single objective expressed in time (Section~\ref{FLoRa Framework}). 

\subsection{Positioning against the state-of-the-art} The latest research in the field of UAV-assisted data collection \textit{independently} addresses a \textbf{subset} of the aforementioned challenges. To name a few, the authors in~\cite{UAV_LoRa_Attention} jointly optimize UAV path-planning and data collection locations to minimize total energy consumption of both UAV and ground sensors while ensuring high packet reception rates. However, their work assumes the following: $i)$ `Always-on' secondary radio to wake up the ground IoTDs, $ii)$ Absence of subset selection: Sufficient UAV battery for the complete UAV tour, $iii)$ Fixed number of probe attempts, and $iv)$ No consideration for data freshness. The authors in~\cite{LoRa_UAV_AoI} optimize both hovering and flight trajectories for LoRa-enabled data collection aimed at minimizing the Age of Information (AoI). However, their model assumes the following: $i)$ Scheduled data collection (i.e., push-based communication) without any random duty-cycling, and $ii)$ Sufficient battery to visit all nodes. Unlike the aforementioned state-of-the-art, in this article, we present a unified framework that simultaneously addresses UAV energy  constraints, multi-SF LoRa communication, probabilistic data collection, path planning (flying) and positioning (hovering), and time-dependent probabilistic probing for (fresh) data collection.

\par To the best of our knowledge, the existing UAV-based works do not jointly solve the freshness-aware data collection problem involving a subset selection, path-planning, probing on randomly available  nodes in a holistic manner. A detailed comparison of existing work positioned against FLoRa is discussed in Section~\ref{sec:Related Work} and depicted in Table~\ref{tab:related_work_comparison}. Moreover, prior methods relax at least one assumption of PTCOP (Table~\ref{tab:related_work_comparison}). FLoRa is evaluated against established baselines and three related methods, EAMO~\cite{LoRa_UAV_AoI}, Safe-TD3~\cite{SafeTD3}, and Li et al.~\cite{OP_DC}, in the PTCOP setting (Section~\ref{sec:Results}). 
 
\par  As FLoRa framework addresses the comprehensive network problem setting of PTCOP in nature, for evaluation purposes, we have considered baseline algorithms such as greedy, metaheuristic (ant-colony optimization), and learning-based variants (such as DRL) that are widely applied in the UAV-based path-planning research. It is worth noting that comparing against baseline approaches for evaluation is a well-established practice and recent works continue to adopt this methodology \cite{Avg_Energy_const_Globecom2025}.

\begin{figure}[t]
    \centering
\includegraphics[scale=.30]{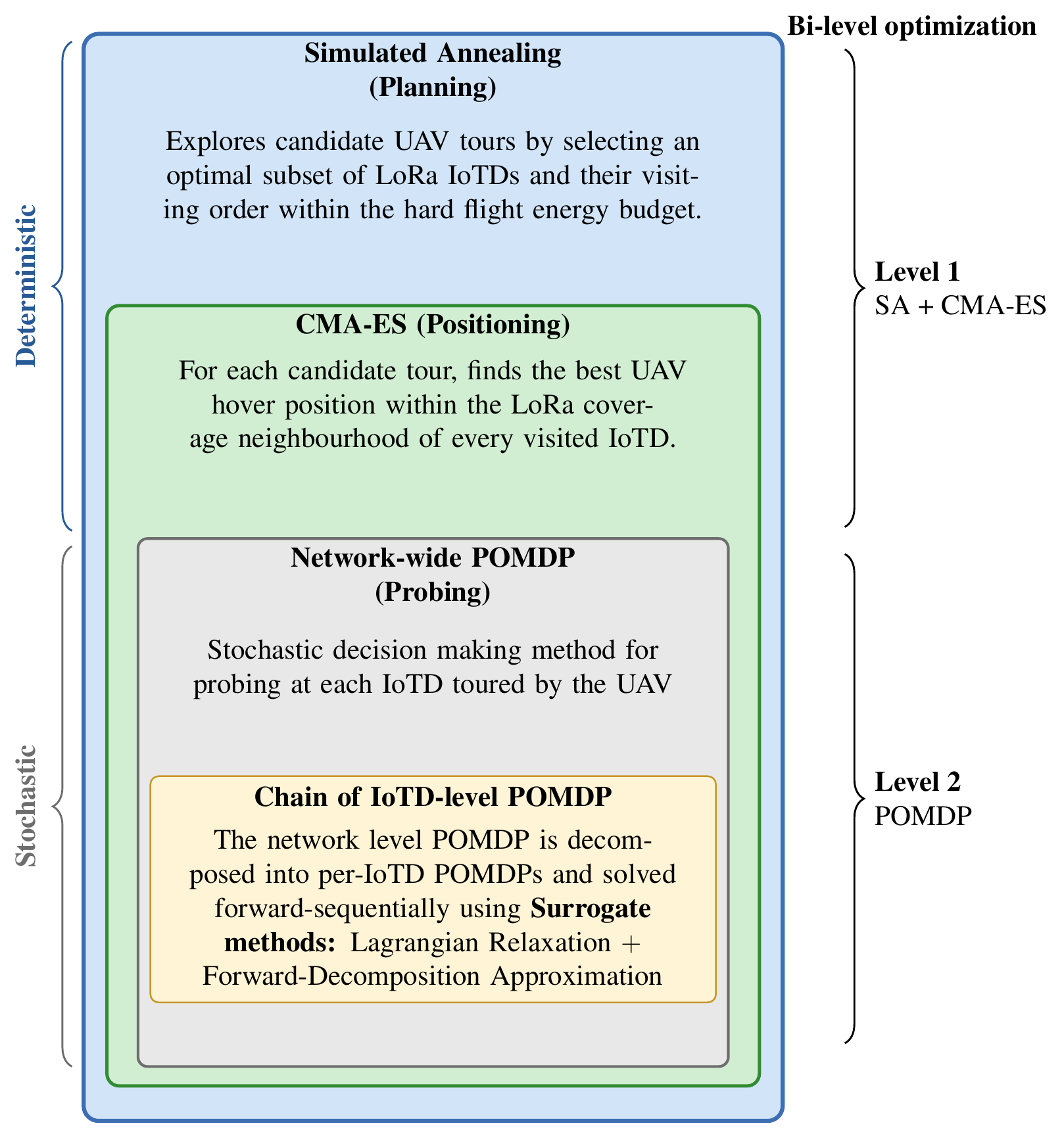}
    \caption{Hierarchical layout of the proposed FLoRa framework} 
\label{fig:floraconceptual}
\end{figure}
\subsection{FLoRa design principles}
The hierarchical structure of the proposed FLoRa framework is depicted in Fig.~\ref{fig:floraconceptual} and its various layers are described as follows:  Simulated Annealing (SA)~\cite{SA_Main},~\cite{SA_Np_Hard} typically used to solve NP-hard combinatorial problems is used to explore the large solution space of path planning aimed at deciding the order of visits to IoTDs. Each candidate path (chosen by SA) is evaluated by invoking Covariance Matrix Adaptation Evolution Strategy (CMA-ES)~\cite{CMA_ES_Tutorial} which in turn will address the \textit{close-enough} problem of finding a suitable hovering position in the coverage neighborhood of each IoTD. Note that the discontinuity of SNR at SF transitions makes the derivative-free nature of CMA-ES more suitable, as its Gaussian smoothing (blurring behavior) acts as a low-pass filter which plays the role of a regularizer for SF-induced discontinuity. The top two layers of SA and CMA-ES consider hard (battery) constraints (without Lagrangian relaxation).

\par Upon hovering on the coverage neighborhood of an IoTD, the UAV faces the following temporal dilemma: To hover and probe the current IoTD or fly to the next IoTD (as per SA/CMA-ES) and probe it? This is a sequential stochastic decision-making problem that requires the UAV to construct a belief using the underlying IoTD's data availability probability. Along with the energy constraints of UAV, this layer of the problem can be effectively modeled using constrained Partially-Observable Markov Decision Process (POMDP)~\cite{bertsekas2012dynamic}, a foundational method for  sequential decision-making in Artificial Intelligence (AI)~\cite{norvigbook}. 

\par \textit{Lagrangian relaxation:} A POMDP needs to maintain the belief states of each Markovian state in a suitable probability simplex. Maintaining hard battery constraints on top of the continuous belief space requires tracking of battery consumption (which is also real-valued) on each sample path.  Tracking the aforementioned battery consumption requires augmenting the state space of POMDP thereby leading to the curse of dimensionality. To this end, for analytical tractability, at the probing layer, we use Lagrangian relaxation approach and convert hard battery constraints to soft (average) constraints problem, where the hovering battery constraints are considered to be met on average. This helps avoid complex state-augmentation method, by only adding a Lagrangian penalty term in the reward function of the POMDP. The Lagrangian battery relaxation also helps each IoTD to independently react to Lagrangian penalty instead of tracking battery consumption on the other nodes.  
\par \textit{Forward decomposition approximation}:
While the global network-shared battery constraints are decoupled (to node-level through penalty) for tractability using Lagrangian relaxation method, this method only transforms constraints into the objective. The global network-shared time resource, however, is a part of the objective (in the form of VIP utility) that needs to be dealt with by a different surrogate method namely, forward decomposition method.

\par The key challenge is that the sequential dependency of UAV's global time cannot be decoupled. In other words, the random times spent by the UAV hovering at preceding IoTDs and deterministic flight time spent by UAV on the tour (made so far) need to be known to the current IoTD towards computing the VIP utility. Keeping track of residual global time in state space (with standard state augmentation method) and tracking in every sample path is also not tractable.  

\par To address the aforementioned challenge, we have developed a forward decomposition method as a surrogate technique to break the dependency at the IoTD level. In this method, at each IoTD, we sequentially solve the underlying IoTD-level POMDP by using the statistics concerned with the expected time spent at preceding IoTDs along with total flight time spent. In this manner, the sample path dependency at the network level is approximated by the average statistics, and we solve IoTD level POMDPs in a forward-sequential manner.

\par \textit{Optimal threshold policy for IoTD-level POMDP}: The decoupled average energy-constrained node-level POMDPs can be solved for optimality. We specifically show the existence of threshold-based optimal policies, which can be computed tractably offline as follows: The threshold-nature of optimal policy induces an absorbing Markov chain, and the expected total Lagrangian constrained reward at each node can be directly computed using the fundamental matrix of the aforementioned absorbing Markov chain. The aforementioned expected reward then serves as a scoring function for the heuristic layer for navigating in the solution space. 

\begin{table*}[t]
\centering

\caption{Comparison with related works}
\label{tab:related_work_comparison}
\setlength{\tabcolsep}{3pt}
\resizebox{\textwidth}{!}{%
\begin{tabular}{|c|c|c|c|c|c|c|c|c|}
\hline
\textbf{Work} & \textbf{Energy} & \textbf{Freshness} & \textbf{Availability} & \textbf{Subset} & \textbf{Probing} & \textbf{Decisions} & \textbf{Method} & \textbf{Evaluated} \\ \hline
\cite{LoRa_UAV_AoI} & Objective & AoI & Scheduled & \texttimes & Fixed & Path, speed, hover & GA + SCA & \checkmark \\ \hline
\cite{SafeTD3} & Constraint & AoI & Always on & Implicit & Single & Trajectory & Constrained DRL & \checkmark \\ \hline
\cite{OP_DC} & Constraint & \texttimes & Always on & \checkmark & Single & Hover points and tour & Greedy & \checkmark \\ \hline
\cite{UAV_LoRa_Attention} & Objective & \texttimes & Always on & \texttimes & Single & Trajectory & Attention DRL & \texttimes \\ \hline
\cite{AttnMORL} & Objective & \texttimes & Always on & Implicit & Single & Trajectory (grid) & Attention MORL & \texttimes \\ \hline
\cite{TWAstar} & \texttimes & AoI & Always on & \texttimes & Single & Hover-point order & Transformer + A* & \texttimes \\ \hline
\cite{UAV_AoI_5} & \texttimes & AoI & Always on & \texttimes & Single & Trajectory and hover & DRL & \texttimes \\ \hline
\cite{TC_UAV_DC_4} & \texttimes & \texttimes & Always on & \checkmark & Single & Trajectory & Exact methods & \texttimes \\ \hline
\cite{UAV_Hover_Flight_Heur} & \texttimes & \texttimes & Always on & \texttimes & Single & Trajectory and hover & Metaheuristics & \texttimes \\ \hline
\cite{UAV_OP_Traj_Opt} & Constraint & \texttimes & Always on & \checkmark & Single & Trajectory & Metaheuristics & \texttimes \\ \hline
\cite{UAV_Subset_Visit_AoI} & Constraint & AoI & Always on & \checkmark & Single & Trajectory and hover & Exact methods & \texttimes \\ \hline
\cite{UAV_Traj_Opt_DC} & \texttimes & \texttimes & Scheduled & \texttimes & Single & Trajectory & Exact methods & \texttimes \\ \hline
\cite{UAV_LoRa_MAB} & Constraint & \texttimes & Always on & --- & --- & Hover & Multi-armed bandits & \texttimes \\ \hline
\textbf{FLoRa} & \textbf{Constraint} & \textbf{VIP} & \textbf{Random} & \textbf{\checkmark} & \textbf{POMDP} & \textbf{Path, hover and probing} & \textbf{Metaheuristics + POMDP} & --- \\ \hline
\end{tabular}}
\end{table*}

\par The key contributions of this work are as follows:
\begin{itemize}
    \item \textbf{VIP metric and PTCOP formulation:} Existing deadline-based freshness metrics assume that a visited IoTD always delivers its data, which makes path planning ill-posed under random duty cycling. VIP penalizes failed probes, which makes the problem well-posed. PTCOP is, to the best of our knowledge, the first formulation that jointly considers subset selection, close-enough hovering over multi-SF LoRa links, randomly available IoTDs, and probing decisions under a hard energy budget.
    \item \textbf{Structural results:} The optimal probing policy at each IoTD is a threshold policy (Theorem~\ref{thm:threshold}), whose expected VIP and expected number of probes have a closed form (Eqs.~\eqref{eq:absexpsteps}--\eqref{eq:expvi}).
    \item \textbf{Tractable decomposition with analysis:} The network-level POMDP is decomposed into IoTD-level POMDPs by Lagrangian relaxation and forward decomposition, with an approximation loss of at most $2.99\%$ (Section~\ref{subsec:Fwd_decomposition}), and the resulting framework is analyzed for convergence and computational complexity (Section~\ref{subsec:convergence}).
    \item \textbf{Evaluation against recent methods:} In addition to greedy, metaheuristic, and DRL baselines, FLoRa is compared with the recent EAMO~\cite{LoRa_UAV_AoI}, Li et al.~\cite{OP_DC}, and Safe-TD3~\cite{SafeTD3}, which it outperforms by $54$--$109\%$, $30$--$77\%$, and $54$--$116\%$ in total expected VIP, respectively.
\end{itemize}
The remainder of this paper is organized as follows: Section~\ref{sec:Related Work}
 reviews related work on UAV-assisted data collection, orienteering problems, and freshness metrics. Section~\ref{sec:System Model} presents the system model, including the data generation, communication, and UAV energy consumption models. Section~\ref{sec: VIP} introduces the VIP metric in detail. Section~\ref{FLoRa Framework} formulates the FLoRa problem and discusses its 
complexity, and Section~\ref{sec: Hierarchical Solution} presents the hierarchical solution framework that decomposes the problem into path planning, 
hover-position optimization, and probing-decision optimization. These three subproblems are then discussed in detail in the subsequent sections: path planning via SA in Section~\ref{sec:SA}, hover-position optimization via CMA-ES in Section~\ref{sec:CMAES}, The surrogate methods, namely Lagrangian relaxation, forward-decomposition approximation along with threshold based optimal policy for average-constrained IoTD-level POMDP are discussed in Section~\ref{sec:POMDP}. The design choice of the proposed FLoRa framework and the purpose of key layers and the benefit of their hierarchical arrangement is discussed in Section~\ref{sec:florawellposedness}. Section~\ref{sec:Results} presents the performance evaluation study comparing FLoRa against baseline methods across different battery budgets and IoTD 
densities. Section~\ref{sec:conclusion} concludes with a discussion on future 
work.

\section{Related Work}\label{sec:Related Work}
Table~\ref{tab:related_work_comparison} summarizes the comparison of related works across the following dimensions: role of energy, freshness metric, IoTD availability, subset selection, probing, decision variables, solution approach, and whether the work is evaluated in Section~\ref{sec:Results}. The proposed FLoRa is positioned against the related works in the \textit{last row of Table~\ref{tab:related_work_comparison}}. In the following, we cluster the survey under path-planning, freshness-based, and LoRa-based data collection, followed by recent works and the research gap.
\subsection{Path-Planning with Energy Constraints}
UAV-enabled data collection has emerged as a promising solution for 
gathering information from distributed IoT networks.  Early research primarily focused on trajectory optimization to 
efficiently visit and collect data from ground sensors. The work 
in~\cite{TC_UAV_DC_4} employs exact methods to optimize UAV trajectories 
for data collection, establishing foundational approaches for path planning 
in UAV-IoT systems.

\par Recognizing that hovering positions significantly impact communication  quality and energy efficiency, subsequent research evolved to jointly  optimize both trajectory and hover locations. The authors  in~\cite{UAV_Hover_Flight_Heur} combine neighborhood search with sampling methods to optimize target subsets and visiting sequences in obstacle-rich environments, maximizing collected rewards within specified travel budgets.  Similarly,~\cite{UAV_OP_Traj_Opt} addresses trajectory optimization under  energy constraints using approximation-based meta-heuristic approaches,  demonstrating the importance of considering UAV battery limitations in  practical deployments. 
\par As shown in Table~\ref{tab:detailed_comparison}, most of the UAV path planning works are push-based in nature. Under duty-cycling network models, the data generation and device availability patterns may not follow deterministic schedules, making push-based approaches impractical. In contrast, pull-based UAV communication enables on-demand data collection wherein the UAVs actively visit and probe ground IoTDs for data without requiring prior synchronization or scheduling. Pull-based methods are generally more energy-efficient as data is transferred only on demand. Our proposed work adopts the pull-based paradigm to deal with the uncertainty introduced by the duty-cycling behavior of the IoTDs.

\subsection{Freshness-Aware Data Collection} 
As applications demanded timely data collection, researchers began incorporating  freshness metrics into UAV data collection frameworks. The concept of Age  of Information (AoI) became central to quantifying data freshness, leading  to optimization problems that balance collection efficiency with information  timeliness. The work in~\cite{UAV_Subset_Visit_AoI} formulates a joint  minimization of AoI and energy consumption by optimizing UAV trajectories  and visiting sequences using exact decomposition methods. Advanced  learning-based techniques have also emerged, as demonstrated  in~\cite{UAV_AoI_5}, which minimizes UAV flight energy subject to AoI  constraints by jointly optimizing the trajectory, hovering position, and  device scheduling using a hierarchical decomposition approach combining  reinforcement learning and convex approximation methods. 
\par Deadline-based metrics such as UoI~\cite{UoI} and VoI~\cite{VoI_with_deterministic_deadline} capture data-collection utility that is suitable for disaster-recovery scenarios. These metrics decay over time and become stale after a specified deadline. The aforementioned deadline-based freshness metrics inherently assume that the underlying data collection device is always `on'. Under duty-cycling devices, the aforementioned metrics do not work during the `off' condition, where the UAV probes multiple times without response, which makes the problem ill-posed. We therefore propose the novel VIP freshness metric for pull-based systems. 

\subsection{LoRa-Based Data Collection} LoRa technology has gained significant attention for UAV-based data 
collection due to its long-range, low-power characteristics, making it 
particularly suitable for wide-area IoT deployments. The authors 
in~\cite{UAV_LoRA_Directivity_Aware} optimize LoRa transmission parameters 
including spreading factor and packet length based on a predetermined UAV 
flight path, using directivity-aware link models to adapt to link quality 
variations in multi-SF scenarios. However, LoRa operates in unlicensed ISM 
bands and is subject to strict regional duty cycle regulations, which limit the fraction of time a device may occupy the channel. This 
introduces an additional coordination challenge between UAV scheduling 
and sensor transmission opportunities. The work in~\cite{UAV_Traj_Opt_DC} 
investigates joint minimization of UAV propulsion energy and sensor energy 
by coordinating trajectory planning with sensor wake-up scheduling, 
highlighting the importance of aligning UAV operations with sensor 
activity constraints.

\par Building on these efforts, several works have addressed LoRa-specific 
challenges in UAV data collection. A learning-based approach 
in~\cite{UAV_LoRa_MAB} enables ground sensors to independently select 
optimal LoRa transmission settings to communicate with a hovering UAV 
gateway under energy constraints. An analytical model 
in~\cite{UAV_LoRa_Agriculture} calculates minimum proximity distances 
between UAVs and sensors for reliable LoRa data collection, balancing 
flight paths and duty cycle constraints to ensure complete coverage within 
battery limits. Joint optimization of UAV trajectory and data collection 
locations to minimize energy while ensuring high packet reception rates 
is addressed in~\cite{UAV_LoRa_Attention}. Combining LoRa capabilities 
with freshness considerations,~\cite{LoRa_UAV_AoI} extends AoI 
optimization to LoRa-enabled deployments by jointly optimizing hovering 
and flight trajectories under scheduled data collection scenarios. As shown in Table~\ref{tab:related_work_comparison}, most state-of-the-art works assume scheduled duty-cycling, or implicitly assume always-on nodes. In this work, we consider more practically relevant random duty-cycling of LoRa nodes for data collection.

\subsection{Recent Works and Research Gap}\label{subsec:gap}
Safe-TD3~\cite{SafeTD3} designs the trajectory with constrained deep reinforcement learning under an energy budget, but assumes always-available devices and AoI without a deadline. TWA*~\cite{TWAstar} orders the hover points with a transformer and A* search, but visits all of them with one transmission each. The multi-objective reinforcement learning of~\cite{AttnMORL} trades collected data against energy on a small grid with always-on devices. EAMO~\cite{LoRa_UAV_AoI} has the most complete LoRa model, but visits all IoTDs and assumes a known duty cycle. Li et al.~\cite{OP_DC} greedily select hover points under an energy budget, without freshness or uncertain availability. 
\par To the best of our knowledge, no existing work combines subset selection under a hard energy budget, close-enough hovering over multi-SF LoRa links, randomly available IoTDs, and adaptive probing with a deadline-based freshness value. FLoRa fills this gap.

\section{System Model}\label{sec:System Model}

\begin{figure}[t]
        \centering
\includegraphics[scale=.3]{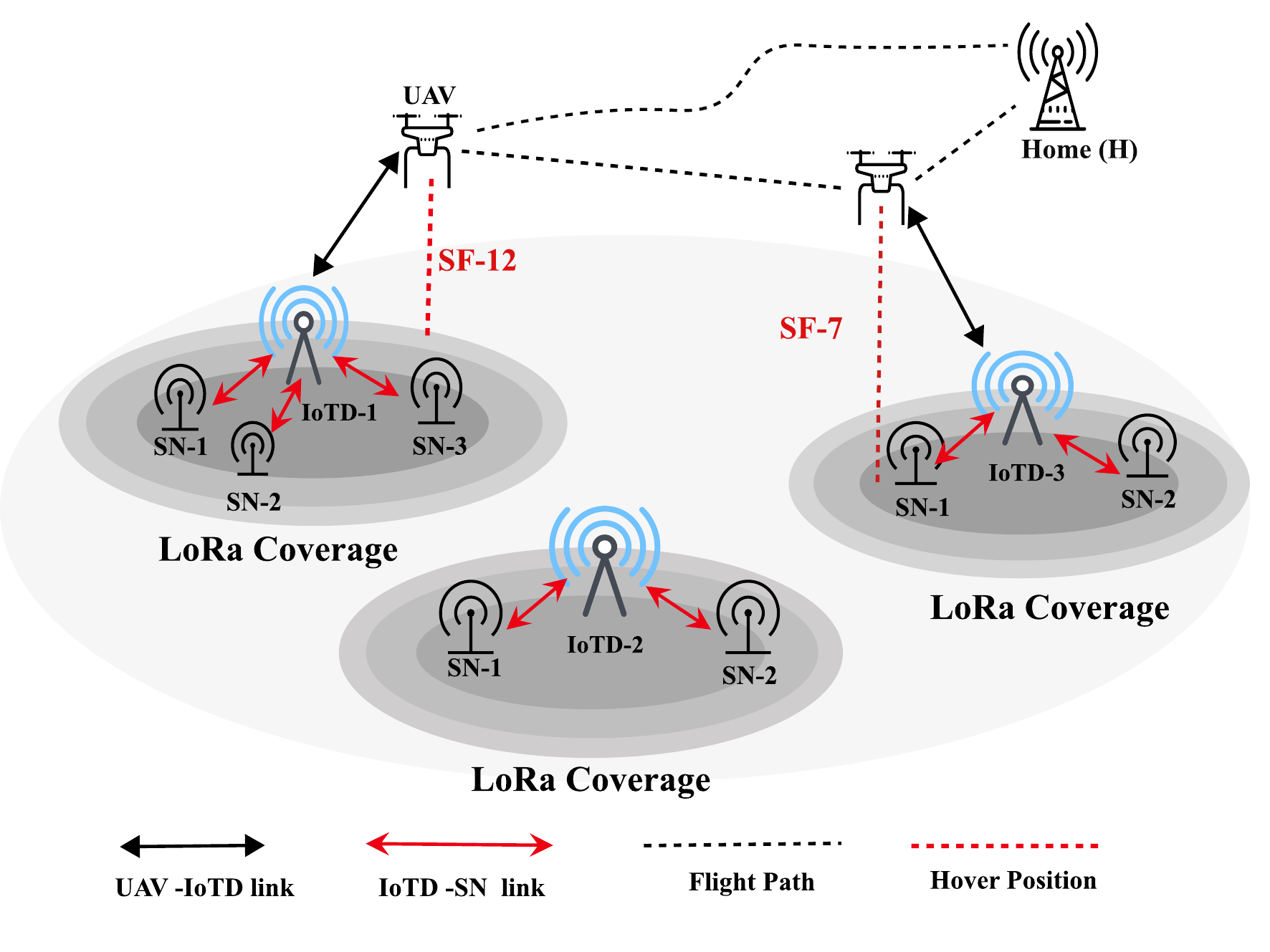}
    \caption{UAV-assisted LoRa IoT data-collection architecture with a sample trajectory over $N=3$ sensing regions, $[N]=\{1,2,3\}$. The UAV visits IoTDs $1$ and $3$ and returns home, giving the tour $\mathbf{z}=\{z_1=1,\,z_2=3\}$.}
    \label{fig: System_Model}
    
\end{figure}

We consider a network of $N$ sensing regions distributed across a wide geographical area, where each region $k \in [N]$ is managed by a LoRa IoTD acting as a local sink. As shown in Fig.~\ref{fig: System_Model}, a UAV equipped with a LoRa gateway flies with constant velocity $v$ at a constant height from a home location $H$, with a finite battery budget $b_{\max}$, to collect data from as many IoTDs as possible. Since the UAV operates under a finite battery budget, it selects a subset of IoTDs and visits them in some order. We represent the tour (ordered set of IoTDs visited) as $\mathbf{z} = (z_i)_{i=1}^{|\mathbf{z}|} \in \mathcal{Z}$, where $z_i$ is the $i^{\text{th}}$ IoTD visited in the tour. Without loss of generality, every tour starts and ends at the home station $0$, which is not explicitly shown in $\mathbf{z}$.
\subsection{Data Generation Model}
At the beginning of each UAV mission, the IoTD in region $k$ (where $k \in [N]$) probes its underlying sensor network and is available for data collection with a known probability $a_k \in [0,1]$. For details on the practical benefits of random duty-cycling, the readers are referred to~\cite{Dronet_Kakarla}. A two-time-scale process is considered, wherein the data-availability of each IoTD is relatively static during the tour time of the UAV. Accordingly, the availability of IoTD $k$ on a given mission is a static latent realization $A_k \sim \text{Bernoulli}(a_k)$, where $A_k$ is a random variable and its outcome is hidden from the UAV (only its Bernoulli parameter $a_k$ is known during configuration time). 
\par We assume that the statistical characterization of the underlying data-generating sensor network is known~\cite{availability_tarak} (e.g., a Bernoulli delivery probability $p_k$), so that $a_k$ is a known function of $p_k$. Even so, the problem remains non-trivial because the realization $A_k$ is unknown. We shall, however, relax the known-statistics assumption in Section~\ref{sec:POMDP} (Proposition~\ref{prop:learn}) and discuss the working of FLoRa when the availability statistics are unknown.

 \par The aforementioned duty-cycling refers to device-based duty-cycling for energy-efficiency. We now focus on LoRa communication-level duty-cycling. In scheduled duty-cycling~\cite{LoRa_UAV_AoI}, the LoRa IoTD wakes up and listens for channel activity in predefined time slots, requiring the UAV collection activity to be synchronized with this schedule. Random duty-cycling \cite{Random_Duty_Cycling} refers to the fact that  LoRa IoTD listens for channel activity only when it has data to transmit. Otherwise, it remains in the sleep (OFF) state. Once it transmits packets it will remain idle for a certain amount of time (say $90\%$ of the time). With our two-time-scale approximation, we assume that during a tour, the transmissions happen within $10\%$ of communication duty-cycling.

\subsection{Communication Model}\label{sec:Comm Model}
The SF determines the chirp rate in LoRa  modulation, where lower SFs support higher data rates over short ranges and higher SFs provide better reliability at longer ranges.  Unlike a fixed gateway, the UAV serves as a mobile gateway whose distance to each IoTD varies with its hover position. Given the  wide-area coverage of LoRa IoTDs (up to $d_{\max}$ km), the UAV  need not hover directly above each IoTD but can collect data from  anywhere within the coverage neighborhood of IoTD $z_i$. We define  the hover location $l_i \in \mathbb{R}^2$ as the UAV's position  when communicating with IoTD $z_i$, where $d_i = \|l_i - c_i\|  \leq d_{\max}$ and $c_i$ is the physical location of IoTD $z_i$.

The UAV selects the appropriate SF based on $d_i$, and each SF carries a distinct SNR sensitivity threshold $\theta_{SF}$. A packet is received if its SNR $\psi$ exceeds $\theta_{SF}$, so the average packet success probability at hover location $l_i$ is $s_i^{l_i}=\mathbb{P}[\psi\ge\theta_{SF}]=\int_{\theta_{SF}}^{\infty} f_\psi(x)\,dx$, where $f_\psi$ is the probability density function of $\psi$. For Rayleigh fading, where the channel power is exponentially distributed, the average packet success probability is given as~\cite{LoRa_Path_Loss_Main},~\cite{LoRa_Pathloss}
\begin{equation}
s_i^{l_i} = \exp\!\left(-\frac{\eta\, \theta_{SF}}{p_t \cdot G(d_i)}\right),
\label{eq:success_prob}
\end{equation}
where $\eta$ is the effective receiver noise power, $\theta_{SF}$ is the SNR threshold of the selected SF, $p_t$ is the transmit power of IoTD $z_i$, and $G(d_i)$ is the distance-dependent channel attenuation function adopted from~\cite{LoRa_Pathloss}. Since $G(d_i)$ decreases with distance and appears in the denominator, the received power $p_t \cdot G(d_i)$, and hence $s_i^{l_i}$, decreases with $d_i$, consistent with Fig.~\ref{fig:snrsuccprb}.
\par Fig.~\ref{fig:snrsuccprb} plots the average packet success probability of Eq.~\eqref{eq:success_prob} against the hover distance. LoRa assigns a discrete SF to each distance band. Within a band, the SNR threshold is fixed while the path loss increases with distance, so the probability decreases. At each band edge, the link switches to a higher SF with a lower SNR threshold, so the probability increases, at the cost of a longer time-on-air per probe. Hence, hovering just inside an SF band can be better than flying closer, and the objective is discontinuous at every band edge, which motivates the derivative-free CMA-ES in Section~\ref{sec:CMAES}. The same dependence on distance is reported for distance-based SF allocation in~\cite[Fig.~1]{LoRa_Pathloss}.

\subsection{UAV Energy Consumption Model}\label{subsec:UAV_Energy}
The UAV consumes energy in two phases during data collection: flight between nodes and hovering at each node for data transfer. The hovering power consumption $h$ is computed as~\cite{survey}:
\begin{equation}
h = \frac{\vartheta}{8}\rho \Lambda S \Omega^3 \zeta^3 + (1+\beta)\frac{W^{3/2}}{\sqrt{2\rho S}}
\label{eq:hover_power}
\end{equation}
where $W$ is the aircraft weight, $S$ is the rotor disk area, $\Lambda$ is the rotor solidity, $\zeta$ is the rotor radius, $\Omega$ is the blade angular velocity, $\vartheta$ is the profile drag coefficient, $\rho$ is the air density, and $\beta$ is the induced power correction factor. The flight power consumption when the UAV travels at velocity $v$ is~\cite{survey}:
\begin{equation}
f = P_0\left(1 + \frac{3v^2}{v_{\text{tip}}^2}\right) + P_i\left(\sqrt{1 + \frac{v^4}{4v_0^4}} - \frac{v^2}{2v_0^2}\right)^{\frac{1}{2}} + \frac{\xi\rho \Lambda S v^3}{2}
\label{eq:flight_power}
\end{equation}
where $v_{\text{tip}}$ is the rotor blade tip speed, $v_0$ is the mean rotor induced velocity in hover, and $\xi$ is the fuselage drag ratio. Let $t_{i-1,i}$ (with $i \in \{1, 2, \ldots, |\mathbf{z}|\}$) denote the UAV's flight time from a position at IoTD $i-1$ to a position at IoTD $i$ in the tour, and $t_i$ is the hover time at a given position at IoTD $i$. The node-wise energy consumption for reaching and probing an IoTD $i$ in the tour $\mathbf{z}$ is therefore given by:
\begin{equation}
B_i = f  t_{i-1,i} + h t_i
\label{eq:battery_node}
\end{equation}
For the first IoTD, $B_1 = ft_{0,1}+ht_1$ (here $0$ represents `\textit{home}'), and to take into account the return flight energy from the last IoTD, we use $B_{|\mathbf{z}|+1}=ft_{|\mathbf{z}|,0}$.

\section{Proposing VIP Metric: Value of Information for Pull-Based Systems }\label{sec: VIP}

\begin{figure}[t]
    \centering
    \includegraphics[scale=0.65]{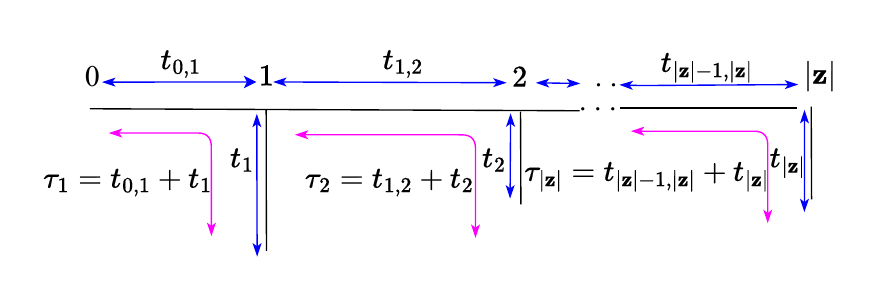}
    \caption{Timing components associated with the UAV flight and hover for each of the IoTDs visited sequentially in a given tour}
    \label{fig:timing}
\end{figure}

Let $T$ denote the common deadline of the data collection mission. Data collected after $T$ has no value, as in disaster response and other time-critical applications. As shown in Fig.~\ref{fig:timing}, the UAV on a given tour $\mathbf{z}$ visits IoTDs (in a specified order) and probes them for data. The total flight and hovering time taken by the UAV after leaving a position at IoTD $i-1$ to the next IoTD $i$ is represented by $\tau_i = t_{i-1,i} + t_i$. For the considered data freshness model, upon visiting an IoTD $i$ and collecting data, the VoI of the data can be represented as a residual freshness as follows: 

\begin{equation}
\Delta_i = 
\begin{cases}
T - \sum_{j=1}^{i} \tau_j & \text{if } \sum_{j=1}^{i} \tau_j \leq T, \\
0 & \text{if } \sum_{j=1}^{i} \tau_j > T.
\end{cases}
\label{eq:voi}
\end{equation}
In other words, the VoI metric at node $i$ is rewritten as, 
\begin{equation}\label{eq:rephrasedvoi}
    \Delta_i = \max\bigg(0,T-\sum_{j=1}^i\tau_j\bigg)
\end{equation}
\par Value of Information (VoI)~\cite{VoI_with_deterministic_deadline} is designed for push-based communication systems where data is transmitted at intervals, and reception is implicitly assumed. In such systems, VoI measures the utility of successfully received data as it decays toward a deadline. However, in pull-based communication systems, the UAV actively probes IoTDs that may not respond due to a duty-cycling `off' state or poor channel conditions, in which case the data may not be collected. The VoI metric is undefined for failed collection attempts, creating a critical gap for optimization-based path planning algorithms that must account for both successful and unsuccessful probes.

To address this limitation, we propose a novel value-based freshness metric called Value of Information for Pull-based systems (VIP). Unlike VoI, which only assigns value to successfully collected data, the VIP explicitly accounts for data collection failures by assigning negative penalties, enabling path planning algorithms to optimize probing strategies while considering the cost of unsuccessful attempts.
Let $V_i$ be the random variable  representing the VIP at  $i^{th}$ IoTD ($i \in \{1, \ldots, |\mathbf{z}|\}$) visited in tour $\mathbf{z}$, where $i$ indexes the tour position and $z_i \in [N]$ is the IoTD visited at that position. The VIP metric is defined as:
\begin{equation}
V_i = \begin{cases} T - \sum_{j=1}^{i} \tau_j & \text{if collected and } \sum_{j=1}^{i} \tau_j \leq T, \\ -\tau_i & \text{otherwise,} \end{cases}
\label{eq:vip}
\end{equation}
where $\tau_i$
penalizes a failed attempt by the time spent at that node. The VIP metric captures the fundamental trade-off in pull-based data collection: successful collection yields positive utility proportional to the remaining
freshness, while failed attempts incur penalties proportional to the
wasted time. This penalty term ensures that optimization algorithms avoid excessive probing at low-probability locations or positions where the residual freshness is already depleted. Since the deadline $T$ is shared by all IoTDs, the time spent at one IoTD lowers the value of every later IoTD. Appendix~\ref{app:freshness} compares VIP with existing freshness metrics.

\par Leveraging VIP, we formulate the UAV path-planning and data-collection problem as a problem of maximizing the total expected VIP across all visited IoTDs, subject to the UAV's limited battery budget. This leads to
a constrained optimization problem over tour selection, hover
positioning, and probing strategies, which we formulate and solve in the following section.

\section{Proposed FLoRa Framework for Fresh Data Collection}\label{FLoRa Framework}

For the aforementioned system model, the  pull-based data collection process aimed at maximizing total expected VIP of the data involves complex sub-tasks of PTCOP as follows: $i)$ Subset selection with path planning (orienteering), $ii)$ Positioning, $iii)$ Probing under probabilistic IoTD availability under limited energy constraints of the UAV. 

\subsection{FLoRa: Problem Formulation}

The key objective of the FLoRa framework is as follows:
\\\textit{``To maximize the total expected VIP freshness of the collected data subject to the hard energy budget constraints of the UAV''}
\\To this end, the optimization formulation of  FLoRa is as follows:
\begin{equation}
\begin{aligned}
\max_{\mathbf{z},\, \mathbf{l},\, \mathbf{n}} \quad & \sum_{i=1}^{|\mathbf{z}|} \mathbb{E}[V_i] \\
\text{subject to} \quad & B_{|\mathbf{z}|+1}+{\sum_{i=1}^{|\mathbf{z}|} B_i \leq b_{\max}}, \\
& d_i \leq d_{\max}, \quad i \in \{1,2,\cdots, |\mathbf{z}|\}, \\
& n_i \in \mathbb{Z}^+\cup \{0\}, \quad i \in \{1,2,\cdots, |\mathbf{z}|\},
\end{aligned}
\label{eq:main_problem}
\end{equation}

where the sum-term $\sum_{i=1}^{|\mathbf{z}|}\mathbb{E}[V_i]$ captures the
total expected VIP (freshness) over the visited IoTDs; $\sum_{i=1}^{|\mathbf{z}|} B_i$
is the total flight-plus-hover energy consumed along tour $\mathbf{z}$, with
$B_i$ defined in Eq.~\eqref{eq:battery_node}. $B_{|\mathbf{z}|+1}$ is the
return-flight energy to the home station; and $b_{\max}$ is the total battery
budget available for the tour. The maximization is over: (i) the tour $\mathbf{z}=(z_i)_{i=1}^{|\mathbf{z}|} \in \mathcal{Z}$, (ii) probing (hovering) positions in the neighborhood of each IoTD, denoted by $\mathbf{l}=(l_i)_{i=1}^{|\mathbf{z}|}$ and (iii) the number of probes at each visited IoTD, denoted by $\mathbf{n}=(n_i)_{i=1}^{|\mathbf{z}|}$. The maximum coverage range at each IoTD is denoted by $d_{\max}$, and $d_i$ is the Euclidean distance from probing location $l_i$ to the IoTD center.
\subsubsection{Problem Characteristics}
The optimization problem involves three coupled decision variables of different types: the tour $\mathbf{z}$ is discrete and combinatorial (leading to a factorial-sized solution space), the hover vector $\mathbf{l}$ is continuous over $\mathbb{R}^{2|\mathbf{z}|}$, and the probe vector $\mathbf{n}$ is integer-valued with $n_i \in \mathbb{Z}^{+} \cup \{0\}$. Furthermore, the expected VIP objective is non-convex and non-differentiable due to the step-like SF transitions in $s_i$ (Fig.~\ref{fig:snrsuccprb}), precluding gradient-based methods. The interplay between these decision variables under stochastic IoTD availability naturally leads to a Mixed-Integer Non-Linear Programming (MINLP) formulation, which makes the problem computationally intractable.
\subsubsection{Competing Effects Captured by the VIP Objective}
Problem~\eqref{eq:main_problem} is a single-objective problem: it maximizes the total expected VIP, while the energy budget and the coverage area are constraints. Coverage, freshness, and data collection probability are competing effects that VIP already combines in a common unit (time). Visiting more IoTDs adds collections but lowers the residual freshness $T-\sum_{j\le i}\tau_j$ of every later IoTD. Flying closer to an IoTD, or probing it longer, raises the collection probability but consumes time, which again lowers the residual freshness of later IoTDs.

\subsubsection{Computational Complexity}
The optimization problem in~\eqref{eq:main_problem} is 
computationally intractable. The tour planning subproblem 
determining which subset of IoTDs to visit and in what order under the battery constraint is an instance of the 
Orienteering Problem~\cite{Orienteering_Problem}, which is NP-hard. Optimizing hover positions over the continuous neighborhood of each IoTD introduces a high-dimensional search space where the objective exhibits multiple local optima. The step-like transitions in packet success probability across LoRa spreading factor levels (Fig.~\ref{fig:snrsuccprb}) create discontinuities that eliminate the possibility of using derivative information. Finally, determining the optimal number of probes $n_i$ at each position involves sequential decision-making under uncertainty, where the UAV observes probe outcomes (success or failure) and must decide whether to probe again or move to the next IoTD without knowing the true availability state of the current IoTD.

\par  In the proposed FLoRa framework, the aforementioned MINLP problem is solved using an ensemble of methods from metaheuristics~\cite{SA_Np_Hard},\cite{meta_heuristics} and sequential decision-making methods~\cite{bertsekas2012dynamic} integrated in a  hierarchical manner. The ensemble of methods solves the aforementioned sub-tasks effectively using a divide-and-conquer approach.

To this end, we adopt a nested decomposition, in which each level optimizes one decision variable while treating the others as fixed parameters. Specifically, the tour $\mathbf{z} \in \mathcal{Z}$ is optimized at the upper level, the hover positions $\mathbf{l}$ are optimized at the middle level for a given tour $\mathbf{z}$, and the probing decisions $\mathbf{n}$ are optimized at the lower level for given $\mathbf{z}$ and $\mathbf{l}$. This decomposition is natural because each subproblem has a fundamentally different structure that suggests a different solution method: the upper level tour optimization is a combinatorial NP-hard problem suited to metaheuristic methods such as SA, the middle level hover optimization is a continuous non-differentiable problem suited to derivative-free evolutionary methods such as CMA-ES, and the lower level probing optimization is a sequential decision-making problem under uncertainty suited to POMDP-based methods.

\section{Hierarchical Solution Framework of FLoRa}\label{sec: Hierarchical Solution}
The tour selection 
over all possible paths, denoted by $\mathcal{Z}$, is NP-hard due to the combinatorial nature of the 
subset selection and ordering problem. The hover position optimization 
over $\mathbb{R}^{2|\mathbf{z}|}$ is non-convex and non-differentiable due to the 
step-like SF transitions in the packet success probability. The probing 
optimization involves sequential decision-making under uncertain IoTD 
availability. The solution requires multiple solvers to address these challenges. We therefore decompose the problem into 
three hierarchical subproblems by defining the following value functions with the notation convention $f(x;y)$ wherein  decision variables (such as $x$) and given parameters (such as $y$) are separated by a semicolon. For brevity, only the parameters that vary during optimization are shown. For instance, parameters such as $d_{\max}$ are not shown (but considered implicitly):

\begin{align}
R(\mathbf{z}) &= \max_{\mathbf{l}} F(\mathbf{l}; \mathbf{z}) \label{eq:R}\\
F(\mathbf{l} ; \mathbf{z}) &= \max_{\mathbf{n}} L(\mathbf{n};\mathbf{z}, \mathbf{l},b_{\max}-b_f) \label{eq:F}
\end{align}

where $b_f = \sum_{i=1}^{|\mathbf{z}|} f\,t_{i-1,i} + f\,t_{|\mathbf{z}|,0}$
denotes the total flight energy of tour $\mathbf{z}$ at hover coordinates
$\mathbf{l}$ (including the return leg to the home station), and the residual
$b_{\max}-b_f$ is the energy available for probing.

\noindent $R(\mathbf{z})$ denotes the optimal  reward achievable for a given tour $\mathbf{z}$ by optimizing hover  positions and probing decisions, and $F(\mathbf{l} ; \mathbf{z})$ denotes the optimal reward for given tour $\mathbf{z}$ and hover positions $\mathbf{l}$ by optimizing probing decisions. The original problem then decomposes into three subproblems:

\begin{align}
\mathbf{z}^* &= \arg\max_{\mathbf{z}} R(\mathbf{z}) \tag{SA} \label{eq:SA}\\
\mathbf{l}^* &= \arg\max_{\mathbf{l}} F(\mathbf{l} ; \mathbf{z}) \tag{CMA-ES} \label{eq:CMAES}\\
\mathbf{n}^* &= \arg\max_{\mathbf{n}} L(\mathbf{n} ; \mathbf{z}, \mathbf{l}, b_{\max}-b_f) \tag{POMDP} \label{eq:POMDP}
\end{align}

\noindent where SA optimizes the tour $\mathbf{z}$, invoking CMA-ES for 
each candidate tour to optimize hover positions $\mathbf{l}$, which in turn invokes POMDP for each candidate hover vector to optimize probing 
decisions $\mathbf{n}$. The proposed hierarchical solution framework is illustrated in Fig.~\ref{fig:Flora_Solution_Framework}.

\par The details of each subproblem are presented in Sections~\ref{sec:SA}, \ref{sec:CMAES}, and 
\ref{sec:POMDP}, respectively. 

\section{P1: Solving Path Planning Using Simulated Annealing Algorithm} \label{sec:SA}\label{SA}

This section describes a metaheuristic approach for solving the Orienteering Problem (OP): selecting a subset of IoTDs $\mathbf{z}\in\mathcal{Z}$ and an order in which the UAV visits them. Finding an optimal tour $\mathbf{z}^*$ is NP-hard. Hence, we resort to near-optimal methods. In particular, Simulated Annealing (SA)~\cite{SA_Np_Hard},~\cite{SA_TOP} is a well-known metaheuristic for combinatorial optimization problems such as OP~\cite{Orienteering_Problem}. Exact solution methods are infeasible at the network scales considered in this work, so we seek high-quality solutions through metaheuristic search.

\textbf{Working principle of SA}: The SA algorithm (shown in Algorithm~\ref{alg:SA}) navigates the large solution space effectively to solve the following path planning problem:
\begin{align}
\mathbf{z}^* &= \arg\max_{\mathbf{z} \in \mathcal{Z}} R(\mathbf{z})
\end{align}

The SA algorithm starts with a candidate path say, $\mathbf{z}_{\tt init}$ and invokes CMA-ES algorithm to solve the subproblem (P2) of finding $\mathbf{l}^*$ which in turn calls POMDP to solve the subproblem (P3) of finding $\mathbf{n}^*$, for the chosen candidate tour $\mathbf{z}_{\tt init}$. At the end of this process, the score of the chosen path is computed from $R(\mathbf{z}_{\tt init})$.

\begin{figure}[t]
    \centering
    \includegraphics[scale=.7]{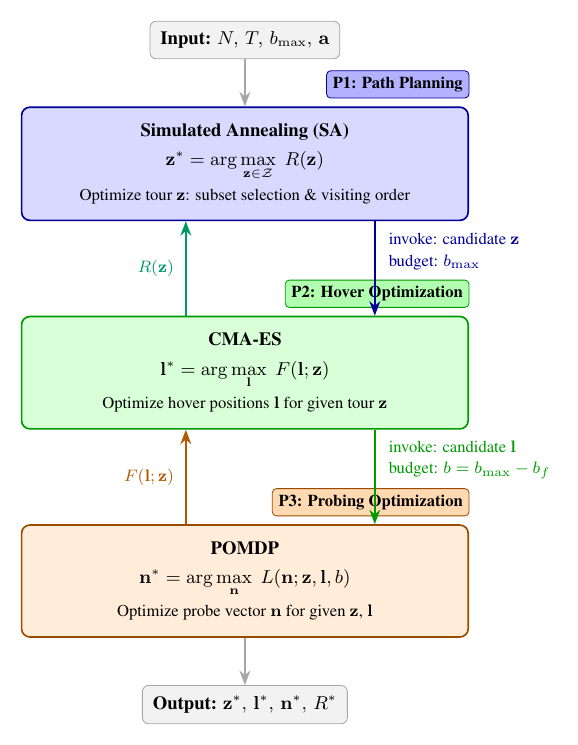}
    \caption{Conceptual block diagram of the proposed FLoRa framework}
    \label{fig:Flora_Solution_Framework}
\end{figure}

\par At each iteration, SA generates a candidate tour $\hat{\mathbf{z}}$ by perturbing the current tour $\mathbf{z}$ with one of four operations, chosen uniformly at random. A swap exchanges the positions of two visited IoTDs, a shift moves a visited IoTD to a different position in the tour, an insert adds an unvisited IoTD to the tour, and a delete removes a visited IoTD from the tour. The candidate is evaluated by using Eq.~\eqref{eq:R}, i.e., by the underlying CMA-ES and POMDP layers. Acceptance and cooling follow the standard SA~\cite{SA_Main} (Appendix~\ref{app:sa_cmaes}), with the score difference $\delta(\hat{\mathbf{z}},\mathbf{z})$ of \textit{Line 10 of Algorithm~\ref{alg:SA}}. The temperature $\Upsilon$ is reduced geometrically as $\Upsilon \leftarrow \alpha \Upsilon$ at each iteration. The algorithm terminates when $\Upsilon$ drops below $\Upsilon_{\min}$ or the maximum number of iterations $K_{\max}$ is reached, with hyper-parameters listed in Table~\ref{tab:sim_params}.

\begin{algorithm}[t]
\caption{SA = $\arg\max_{\mathbf{z}}R(\mathbf{z})$}
\label{alg:SA}
\begin{algorithmic}[1]
\State \textbf{Global Input:} $[N]$ (set of IoTDs), $T$ (deadline),  $\mathbf{a}$ (availability prob. vector), $b_{\max}$ (max energy budget).
\State \textbf{Output:}  $\mathbf{z}^*$ (tour)
\State \textbf{Auxiliary Output}: $\mathbf{l}^*$, $\mathbf{n}^*$
\State \textbf{Initialize:} $\mathbf{z} \leftarrow $ random permuted subset of N, $\alpha=0.995$ (cooling rate) ,\;
        $\Upsilon \leftarrow \Upsilon_0 $ (temperature, $> \Upsilon_{\min}$), $\mathbf{z}^*\leftarrow \mathbf{z}$, itr$=0$
\State CurrentScore $\leftarrow\max_{\mathbf{l}}F(\mathbf{l};\mathbf{z})$
\State BestScore  $\leftarrow$ CurrentScore
\While{$\Upsilon > \Upsilon_{\min}$ and $itr < K_{max}$}
    \State $\hat{\mathbf{z}} \leftarrow \textbf{perturb}(\mathbf{z})$
           $\triangleright$ via \{swap, shift, insert, delete\} 
   
   \State NewScore$(\hat{\mathbf{z}}) \leftarrow \max_{\mathbf{l}}F(\mathbf{l};\hat{\mathbf{z}})$
    \State $\delta(\hat{\mathbf{z}},\mathbf{z})\leftarrow$ NewScore - CurrentScore
\If {$\delta(\hat{\mathbf{z}},\mathbf{z}) > 0$ or $\exp\frac{\delta(\hat{\mathbf{z}},\mathbf{z})}{\Upsilon}>$random$(0,1)$}
            \State $\mathbf{z}\leftarrow \hat{\mathbf{z}}$
            \State CurrentScore $\leftarrow$ NewScore
            \If {CurrentScore > BestScore}
            \State $\mathbf{z}^* \leftarrow \mathbf{z}$
            \State BestScore $\leftarrow$ CurrentScore
        \EndIf
    \EndIf
    \State $\Upsilon \leftarrow \Upsilon \times \alpha$
    \State $itr \leftarrow itr+1$
\EndWhile
\State \Return $\mathbf{z}^*$
\end{algorithmic}
\end{algorithm}

\section{P2: Solving  Positioning Problem  Using CMA-ES Algorithm}\label{sec:CMAES}
For a given candidate tour $\mathbf{z}$ by the SA algorithm (discussed in the previous section), we find the best hovering positions, subproblem (P2), at each of the neighborhoods of IoTDs visited in the candidate tour $\mathbf{z}$.  To this end, the optimal hovering positions problem is formulated as the following \textit{continuous optimization} problem:
\begin{align}
    \mathbf{l}^* &= \arg\max_{\mathbf{l}} F(\mathbf{l} ; \mathbf{z}) \label{eq:CMAES}
\end{align}
\textbf{P2 problem characteristics:}~As described by Eq.~\eqref{eq:success_prob} and the same visualized in Fig.~\ref{fig:snrsuccprb}, the SNR surface as a function of Euclidean distances is $i)$ piecewise continuous, $ii)$ several ridges of sharply changing SNR values at the borders of SF regions, and $iii)$ gradually ranging SNR values within the SF regions. The step-like SNR transitions restrict the use of gradient-based optimization models, making metaheuristics-based methods suitable for such problems. 

To solve such a problem, we need a `variable metric' method that learns the shape of the terrain  and recognizes noticeable changes (ridges) of favorable regions in the terrain. In such a setting involving narrow ridges in the solution space, the SA algorithm yields ineffective results as it is primarily dependent on the random perturbation method that most likely chooses bad solutions (off the ridges), requiring more time to converge.
\textbf{Need for CMA-ES:} CMA-ES~\cite{CMA_ES_Tutorial} is particularly well-suited for solving the subproblem P2. CMA-ES initially starts with a circular-shaped region around each IoTD and uses the covariance matrix update to learn a narrow ridge shape and stretch/rotate it into a suitable ellipsoid to match the narrow ridges. In this manner, the evolutionary and metaheuristic approach of CMA-ES most effectively finds the position at each IoTD from the given SA candidate tour of $\mathbf{z}$.

\par CMA-ES samples candidate solutions from a Gaussian distribution $\mathcal{N}(\mathbf{m},\sigma^2\mathbf{C})$ with mean $\mathbf{m}$, step size $\sigma$, and covariance matrix $\mathbf{C}$. Sampling, selection, and adaptation follow the standard CMA-ES~\cite{CMA_ES_Tutorial} (Appendix~\ref{app:sa_cmaes}).

For our setup, the 
search space is $\mathbb{R}^{2|\mathbf{z}|}$, where each candidate solution 
is a hover vector $\mathbf{l} = (l_1, \ldots, l_{|\mathbf{z}|})$. 
The distribution is initialized with mean $\mathbf{m}_0$ set to 
the IoTD center coordinates $\{c_i\}_{i=1}^{|\mathbf{z}|}$, initial step 
size $\sigma_0$, and covariance $\mathbf{C}_0 = \mathbf{I}_{2|\mathbf{z}|}$. 
The fitness function for a candidate hover vector is defined as:
\begin{equation}
F(\mathbf{l} ;\mathbf{z}) = \max_{\mathbf{n}} L(\mathbf{n};\mathbf{z}, \mathbf{l}, b_{\max}-b_f) 
\end{equation}
\begin{algorithm}[t]
\caption{CMA-ES = $\max_{\mathbf{l}}F(\mathbf{l};\mathbf{z})$}
\label{alg:CMAES}
\begin{algorithmic}[1]
\State \textbf{Global Input:} $[N]$ (set of IoTDs), $T$ (deadline),  $\mathbf{a}$ (availability prob. vector), $b_{\max}$
\State \textbf{Parameter Input:} $\mathbf{z}$ (tour), $b_{\max}$ (battery budget)
\State \textbf{Output:} $\arg\max_{\mathbf{l}}F(\mathbf{l};\mathbf{z})$, $\max_{\mathbf{l}}F(\mathbf{l};\mathbf{z})$
\State \textbf{Initialize:}
       $g=0$, $\mathbf{m}_g \leftarrow  (c_1, \ldots, c_{|\mathbf{z}|})$ (center coordinates),
       $\sigma_g$,
       $\mathbf{C}_g \leftarrow \mathbf{I}_{2|\mathbf{z}|}$, $g_{\max}=$ a function of $N$
\Repeat
    \State Sample $\Gamma$ candidates
           $\{\mathbf{l}_j\}_{j=1}^{\Gamma} \sim
           \mathcal{N}(\mathbf{m}_g, \sigma_g^2 \mathbf{C}_g)$
    \For{$j = 1$ to $\Gamma$}
    \If{$b_f> b_{\max}$ or $t(\mathbf{z},\mathbf{l}_j)>T$}
         \State $\Gamma \leftarrow \Gamma \setminus \{\mathbf{l}_j\}$
         \State \textbf{continue}
        \EndIf
        \State Score$(\mathbf{l}_j) \leftarrow
               \max_{\mathbf{n}}L(\mathbf{n};\mathbf{z}, \mathbf{l}^{(j)}, b_{\max}-b_f)$
    \EndFor
    \State sort($\{\mathbf{l}_j\}_{j=1}^{\Gamma}$) by Score($\mathbf{l}_j$)  and select top $\mu = \lfloor\Gamma/2\rfloor$ scored candidates 
    \State Use top $\frac{1}{2}$ sorted candidates $\mu$ to update $\mathbf{m}_{g+1}$, $\mathbf{C}_{g+1}$,
           $\sigma_{g+1}$ $\triangleright$ (as per ~\cite{CMA_ES_Tutorial})
    \State $g \leftarrow g + 1$ $\triangleright$ next generation sampling.
   
\Until{$\sigma_g < 10^{-6}$ \textbf{or} $g > g_{\max}$}
\State $\mathbf{l}^* \leftarrow \mathbf{m}_g$
\State Score$(\mathbf{l}^*) \leftarrow
       \max_{\mathbf{n}}L(\mathbf{n};\mathbf{z}, \mathbf{l}^*, b_{\max}-b_f)$
\State \Return $\mathbf{l}^*$, Score$(\mathbf{l}^*)$
\end{algorithmic}
\end{algorithm}

\section{P3: Finding Optimal Number of Probes using POMDP Approach }\label{sec:POMDP}
The key challenge in $P3$ is that the underlying POMDP is inherently stochastic, operates over a continuous belief space and is also typically PSPACE-hard. It is critical to note that the upper heuristic layer of CMA-ES (as described in $P2$) needs to invoke POMDP \textit{multiple times} for score (fitness value) evaluation in an iterative manner. Most importantly, the strategy for the entire path, positioning, and probing needs to be computed offline before the tour starts. To this end, analytical tractability plays a central role over efficiency. While the deterministic nature of the  SA and CMA-ES heuristic layers helps in solving the problem tractably with the  hard constraints of battery availability $b_{\max}$,  accounting for such hard battery constraints in every sample path of the POMDP is computationally intractable (as it requires augmenting state space with continuous battery levels). For analytical tractability, we therefore approximate the remaining hard battery constraint $b_{\max}-b_f$ as an expected constraint with the help of Lagrangian relaxation. 
\par From the upper heuristics layers of SA and CMA-ES, we get the tour $\mathbf{z}$ and positioning coordinates $\mathbf{l}$, respectively. After deducting battery consumed during the UAV's flight $b_f$ from the maximum available battery $b_{\max}$, with the remaining budget of $b=b_{\max}-b_f$ the goal of POMDP is formulated as follows:
\begin{equation}
    \max_{\mathbf{n}}L(\mathbf{n};\mathbf{z}, \mathbf{l},b)
\end{equation}
\subsection{Approximation 1: Lagrangian Relaxation}\label{sc:approx1lagrang}
As discussed earlier, to avoid the curse of dimensionality in tracking the battery in the state space (using state augmentation), we alternatively consider expected battery constraints as follows:
\begin{align}
     \max_{\mathbf{n}}L(\mathbf{n};\mathbf{z}, \mathbf{l},b) &= \max_{\mathbf{n}}\sum_{i=1}^{|\mathbf{z}|}\mathbb{E}[V(n_i)]\\
     {\tt s.t.} \quad \sum_{i=1}^{|\mathbf{z}|}\mathbb{E}[B(n_i)]  &\le b
\end{align}
where $\mathbf{n}=(n_i)_{i=1}^{|\mathbf{z}|}$ is the ordered tuple of the number of probes made by the UAV at each IoTD $i$, $V(n_i)=V_i$ is the random variable representing VIP metric at IoTD $i$ which is a function of number of probes $n_i$, and $B(n_i)$ is the hovering battery consumption, which is also a function of the time taken  by the number of probes $n_i$ at the IoTD $i$.

\par We can now employ the standard method of converting the aforementioned constrained network-level POMDP into an unconstrained formulation using the Lagrangian function (using Lagrange multiplier, $\lambda \geq 0$) as follows:
\begin{equation}
L(\mathbf{n}, \lambda;\mathbf{z}, \mathbf{l},  b ) = \sum_{i=1}^{|\mathbf{z}|} \mathbb{E}[V(n_i)] - \lambda \sum_{i=1}^{|\mathbf{z}|} \mathbb{E}[B(n_i)] + \lambda b.
\label{eq:lagrangian}
\end{equation}

\par The Lagrange multiplier $\lambda$ in Eq.~\eqref{eq:lagrangian} acts as a shadow price, i.e., the VIP reward earned per unit of battery.  The Lagrangian dual function $D(\lambda;b)$ for Eq.~\eqref{eq:lagrangian} is represented as follows:
\begin{equation}
D(\lambda;b) = \sup_{\mathbf{n}} L(\mathbf{n};\mathbf{z}, \mathbf{l},\lambda).
\label{eq:dual_function}
\end{equation}
\par  The dual problem seeks to find the optimal Lagrange multiplier by solving the optimization formulation: 
\begin{equation}\label{eq:dualproblem}
\inf_{\lambda \geq 0} D(\lambda;b) 
\end{equation}
The aforementioned dual problem can be solved using the bisection method (as shown in Algorithm~$\ref{alg:bisection}$), to get the optimal Lagrange multiplier as follows:
\begin{equation}\label{eq:dualproblemreturn}
\lambda^* = \arg\inf_{\lambda \geq 0} D(\lambda;b).
\end{equation}

\begin{algorithm}[t]
\caption{Bisection Method: $\arg\inf_{\lambda}D(\lambda; b)$}
\label{alg:bisection}
\begin{algorithmic}[1]
\State \textbf{Global Input:} $h$, $\epsilon$
\State \textbf{Parameter Input:} $b$
\State \textbf{Output:} $\lambda^*$
\State Initialize $\lambda_{\text{low}} \gets 0,\; \lambda_{\text{high}} \gets \frac{1}{h}$
\While{$(\lambda_{\text{high}} - \lambda_{\text{low}}) > \epsilon$}
    \State $\lambda_{\text{mid}} \gets \frac{\lambda_{\text{low}} + \lambda_{\text{high}}}{2}$
    \State $b_{\tt mid}=\sum_{i=1}^{|\mathbf{z}|}\mathbb{E}[B(n_i)]$, from the optimal solution of Eq.~\eqref{eq:dual_function}, $\sup_{\mathbf{n}}{L}(\mathbf{n};\mathbf{z}, \mathbf{l},  \lambda_{\tt mid})$, using $\lambda_{\tt mid}$ 
    \If{$b_{\tt mid}> b$}
        \State $\lambda_{\text{low}} \gets \lambda_{\text{mid}}$
    \Else
        \State $\lambda_{\text{high}} \gets \lambda_{\text{mid}}$
    \EndIf
\EndWhile
\State $\lambda^* \gets \lambda_{\text{mid}}$
\State \Return $\lambda^*$
\end{algorithmic}
\end{algorithm}

\par The optimal number of probes at each IoTD $\mathbf{n}$ is computed using the optimal Lagrange multiplier $\lambda^*$ yielded by the bisection method:
\begin{equation}
    \mathbf{n}^* = \arg\sup_{\mathbf{n}}\ L(\mathbf{n}; \mathbf{z},
    \mathbf{l},\lambda^*).
\end{equation}
\subsection{Approximation 2: Forward Decomposition}\label{sc:approx2decoup}
In the previous Section~\ref{sc:approx1lagrang}, the global shared hard energy constraint (ill-posed) was approximately transformed into a soft (expected energy constraint) constraint by using Lagrangian relaxation as regularization method thereby making the problem well-posed.

\par In this section, we focus on the curse of dimensionality of solving the network-wide POMDP optimization problem directly  as depicted in Eq.~\eqref{eq:dualproblemreturn}, wherein the combinatorial combination of each node's states increases exponentially, leading to the curse of dimensionality. For tractability, we approximately decompose the network level POMDP to node-level POMDPs. This second level of approximation is aimed at decoupling the dependence created by the global shared time resource encoded in the  VIP metric $V_i$ (as defined in Eq.~\eqref{eq:vip}) independently across the nodes. The objective of each IoTD $i$ depends on the total time incurred by UAV during flight and hovering on the previous nodes $j<i$, $\tau_j=t_{j-1,j}+t_j$. While the flight times such as $t_{j-1,j}$ are determined at the upper heuristics layer, the POMDP controls the hover time as a function of number of probes $n_j$ at each of the previous nodes $j$.

\par We use the definition of VIP metric from function of $n_i$ to hovering time $t_i=n_i\times t^l_i$ over $n_i$ probes along with hovering energy consumption also as a function of $t_i$ and accordingly unwind the corresponding terms of the  Lagrangian function (in Eq.~\eqref{eq:lagrangian}) to reveal the dependency of shared global time resource $T$ among the IoTDs of the specified UAV path: 
\begin{equation}\label{eq:unwindlagrangian}
L(\mathbf{n},\lambda;\mathbf{z}, \mathbf{l} ) = \mathbb{E}
[V(T,\tau_1)-\lambda B(t_1)+V(T,\tau_1,\tau_2)-\lambda B(t_2)+\cdots].
\end{equation} 
\par From Eq.~\eqref{eq:unwindlagrangian}, as each VIP term $V(.)$ depends on total probing time from previous nodes, and deterministic flight times from previous nodes, the resulting POMDP needs to keep track of global time as a  part of state which leads to curse of dimensionality. We therefore approximate using the following decomposition method by replacing random variables $\tau_i$ with their respective expectations $\mathbb{E}_{n_i}[\tau_i]$ as follows:
\begin{align}\label{eq:decouplelagrangian}
L(\mathbf{n}, \lambda;\mathbf{z}, \mathbf{l}) = \mathbb{E}
[\underbrace{V(T,\tau_1)-\lambda B(t_1)}_{L_1(n_1,\lambda;z_1,l_1)}+\underbrace{V(T,\mathbb{E}_{n_1}[\tau_1],\tau_2)-\lambda B(t_2)}_{L_2(n_2,\lambda;z_2,l_2)}+\cdots].
\end{align} 
\par Note that the decomposition approximation method of replacing the hovering time of previous nodes (a random variable) by its expectation helps solve the node-level problems sequentially, from IoTD$_1$ through to IoTD$_{|\mathbf{z}|}$. In this manner, each subsequent IoTD can tractably use the expected time computed from previous IoTDs in an incremental manner. Now the IoTD level POMDP dual problem  can be solved tractably as follows:

\begin{align}\label{eq:pomdpconv}
D_i(\lambda)=\max_{n_i} L_i(n_i ; z_i, l_i, \lambda) = \nonumber\\
\max_{\pi\in\Pi}\mathbb{E}_\pi\bigg[ \mathbb{I}_{\{\textrm{data~collected}\}}\underbrace{\Big(T-\sum_{j=1}^{i-1} \mathbb{E}_{n_j}[\tau_j]-t_{i-1,i}\Big)}_{\textrm{constant~w.r.t.~}\pi}-  \sum_{t=1}^\infty n_t^\pi t_i^l- \nonumber \\ \lambda\sum_{t=1}^\infty n_t^\pi ht_i^l\bigg]
\end{align}

where $\mathbb{I}\{\text{data~collected}\}$ indicates IoTD $i$ successfully collected data and the policy $\pi$ is defined as $\pi = (n_t^\pi)_{t=0}^\infty$ with  $n_t^\pi \in \{0,1\}$, and $n_i = \sum_{t=1}^\infty n_t^{\pi}$. The summation in Eq.~\eqref{eq:pomdpconv} represents the number of probes made by the UAV at IoTD $i$, with each time slot $t$ of length equal to the SF-dependent LoRa time-on-air (ToA) of a probe.

\par Eq.~\eqref{eq:pomdpconv} represents the stochastic decision process formulation of the optimization sub-problem (P3).  As (the realization of) the state of availability of IoTD $i$ is unknown (hidden) at the UAV, the aforementioned stochastic decision process is formulated as a Partially Observable 
Markov Decision Process (POMDP).

\par \textbf{POMDP Formulation}:
The POMDP is defined by the tuple $(\mathcal{X}, \mathcal{U}, \mathcal{O}, \mathcal{P}, \mathcal{R})$, described as follows.

\textit{State Space} $\mathcal{X}$: The true availability of IoTD $i$ during the mission is fixed but unknown to the UAV. The UAV therefore maintains a belief $x_t \in \mathcal{X} = [0,1]$, where $x_t$ is the posterior probability that IoTD $i$ is available given the probe outcomes observed up to slot $t$. The belief is initialized at the known prior availability probability, $x_1 = a_i$.

\textit{Action Space} $\mathcal{U}$: At each decision epoch $t$, the UAV selects one of three actions: idle ($u_t = 0$, remain hovering without probing), probe ($u_t = 1$, send a probe packet to IoTD $i$), or terminate ($u_t = 2$, stop hovering and depart, i.e., fly to the next IoTD in the tour or return to the home station $0$). Thus $\mathcal{U} = \{0, 1, 2\}$, with $u_t \in \mathcal{U}$. The inter-node flight energy incurred by the terminate action is accounted for at the upper CMA-ES/SA layer; at the node level it is therefore treated as a costless, terminating action leading to an absorbing state $\phi$.

A successful collection requires both that IoTD $i$ is available and that the frame is received over the channel; hence the per-probe success probability is $x_t\, s_i^{l_i}$, where $s_i^{l_i}$ is the packet success probability at hover location $l_i$ from Eq.~\eqref{eq:success_prob}.

\textit{Transition Model} $\mathcal{P}$: Starting from the prior $x_1 = a_i$, the belief evolves under each action as follows. For probe ($u_t = 1$), a successful response confirms and consumes the data, so the node-level process transitions to the absorbing state $\phi$; a failed probe leaves the device belief Bayes-updated as follows:
\begin{align}
    x_{t+1} = \begin{cases}\phi & \text{w.p. } x_t s_i^{l_i}, \\
     \frac{x_t(1 - s^{l_i}_i)}{x_t(1 - s^{l_i}_i) + (1 - x_t)}  &\text{w.p. } (1 - x_t s_i^{l_i}). 
    \end{cases}
    \label{eq:bayesupdate}
\end{align}
For idle ($u_t = 0$), no probe is sent and the belief is unchanged:
\begin{align}
    x_{t+1} &= x_t & \text{w.p. } 1. \label{eq:idletransition}
\end{align}
For terminate ($u_t = 2$), the node-level process ends and transitions to the absorbing state:
\begin{align}
    x_{t+1} &= \phi & \text{w.p. } 1. \label{eq:terminatetransition}
\end{align}

\par \textit{Observation Model}: The probability of observing a successful data retrieval upon probing at IoTD $i$ is:
\begin{align}
\mathbb{P}(o_t = 1 \mid x_t) &= x_t s_i^{l_i}, \\
\mathbb{P}(o_t = 0 \mid x_t) &= 1 - x_t s_i^{l_i}.
\end{align}

\par \textit{Stage Reward} $\mathcal{R}$: The objective in Eq.~\eqref{eq:pomdpconv} is solved with the above POMDP using the following single-stage reward $R_t$ (a random variable). For the probe action ($u_t = 1$):
\begin{align}\label{eq:stagereward}
R_t = \begin{cases}
\left(T-\sum_{j=1}^{i-1}\mathbb{E}_{n_j}[\tau_j]-t_{i-1,i}\right) - t^l_i(1 + \lambda h) & \text{w.p. } x_t s^{l_i}_i, \\[4pt]
-t^l_i(1 + \lambda h) & \text{w.p. } 1 - x_t s^{l_i}_i,
\end{cases}
\end{align}

For the idle action ($u_t = 0$), the UAV pays the hovering cost without probing:
\begin{align}
R_t = -t^l_i(1 + \lambda h) \quad \text{w.p. } 1 \quad (u_t = 0,\ \text{idle}). \nonumber
\end{align}
For the terminate action ($u_t = 2$), the node-level process ends at zero cost (the departure flight energy being accounted upstream):
\begin{align}
R_t = 0 \quad \text{w.p. } 1 \quad (u_t = 2,\ \text{terminate}).
\end{align}

The $-\tau_i$ penalty of Eq.~\eqref{eq:vip} splits across layers: the flight part $t_{i-1,i}$ is charged at the SA/CMA-ES layers (via $b_f$ and the freshness constant $T-\sum_{j<i}\mathbb{E}_{n_j}[\tau_j]$), so the node-level reward in Eq.~\eqref{eq:stagereward} retains only the priced hover cost $t^l_i(1+\lambda h)$.

\begin{lemma}\label{lem:reducedpolicy}[Reduced Policy and Action Space]
Let the per-slot action space be $\mathcal{U}=\{0,1,2\}$, where $0$ denotes idle (remain hovering without probing), $1$ denotes probe, and $2$ denotes terminate (departure to the next IoTD in the tour, or return to the home station $0$). Then there exists a reduced policy space $\tilde{\Pi}$ containing an optimal policy, wherein each $\tilde{\pi}\in\tilde{\Pi}$ has the form
$(\underbrace{1,1,\ldots,1}_{n^{\tilde{\pi}}=\sum_{t}n_t^{\tilde{\pi}}},\,2)$:
the UAV probes for some number of consecutive attempts and then terminates. Equivalently, the idle action is dominated, so the effective optimal decision at every slot reduces to probe versus terminate.
\end{lemma}

\begin{proof}
See Appendix~\ref{app:lemma1}.
\end{proof}

\begin{lemma}\label{lem:Monotone} The state $x_t$ of the POMDP is  monotonically non-increasing over time slots $t$.
\end{lemma}

\begin{proof}
See Appendix~\ref{app:lemma2}.
\end{proof}

\begin{theorem}[Threshold Policy]
\label{thm:threshold}
The optimal probing policy for IoTD $i$ is a threshold policy: there exists a
sequence of thresholds against which the UAV probes at slot $t$ if and only if
\begin{equation}
\label{eq:threshold}
x_t\,s^{l_i}_i \;\ge\; \frac{t^l_i(1+\lambda h)}{T-\sum_{j=1}^{i-1}\mathbb{E}_{n_j}[\tau_j]-t_{i-1,i}},
\end{equation}
and stops otherwise.
\end{theorem}
\begin{proof}
We use finite-horizon backward induction~\cite{bertsekas2012dynamic}. Each
probe consumes a slot of hover time, so only finitely many probes yield a
positive reward; let $\overline{n}$ be the last such slot, with terminal value
$J^*_{\overline{n}}(x_{\overline{n}})=0$ for all $x_{\overline{n}}$. Throughout
probing $T-\sum_{j=1}^{i-1}\mathbb{E}_{n_j}[\tau_j]-t_{i-1,i}>0$, since the policy
terminates once freshness is exhausted.

By Lemma~\ref{lem:reducedpolicy} the idle action ($u_t=0$) is dominated, so on
the optimal path the only non-probe action is terminate ($u_t=2$), which
ends the process at IoTD $i$ at zero cost. Hence the stop branch takes value
$0$, and the stage-indexed Bellman equation is
\begin{equation}
\label{eq:bellman}
\begin{split}
J^*_t(x_t)=\max\Big\{\, 0,\;
-t^{l_i}(1+\lambda h) + x_t s^{l_i}_i \Big(T-\!\sum_{j=1}^{i-1}\mathbb{E}_{n_j}[\tau_j] -t_{i-1,i}\Big)\\
+ (1-x_t s^{l_i}_i)\,J^*_{t+1}(x_{t+1}) \Big\},
\end{split}
\end{equation}
where a failed probe updates $x_{t+1}$ by Eq.~\eqref{eq:bayesupdate}. Probing is
optimal at slot $t$ iff the probe branch is nonnegative.

Base step: At $t=\overline{n}-1$, since $J^*_{\overline{n}}\equiv 0$, the probe
branch equals
$-t^{l_i}(1+\lambda h)+x_t s^{l_i}_i\big(T-\sum_{j=1}^{i-1}\mathbb{E}_{n_j}[\tau_j]-t_{i-1,i}\big)$,
which is nonnegative iff
$x_t s^{l_i}_i \ge t^{l_i}(1+\lambda h)/\big(T-\sum_{j=1}^{i-1}\mathbb{E}_{n_j}[\tau_j]-t_{i-1,i}\big)$,
giving Eq.~\eqref{eq:threshold}. Thus $J^*_{\overline{n}-1}(x_t)\ge 0$, and equals
$0$ whenever the condition fails.

Induction step: Assume that at slot $t+1$ probing is optimal iff
Eq.~\eqref{eq:threshold} holds, that $J^*_{t+1}\ge 0$, and that $J^*_{t+1}=0$
whenever Eq.~\eqref{eq:threshold} fails. Consider slot $t$ with
$x_t s^{l_i}_i < t^{l_i}(1+\lambda h)/\big(T-\sum_{j=1}^{i-1}\mathbb{E}_{n_j}[\tau_j]\big)$.
The right-hand side of Eq.~\eqref{eq:threshold} depends only on the flight and
hover times of the preceding nodes $j<i$ and is therefore constant across the
probe slots of IoTD $i$. By Lemma~\ref{lem:Monotone}, a failed probe gives
$x_{t+1}\le x_t$, so $x_{t+1}s^{l_i}_i \le x_t s^{l_i}_i$ still violates the
(unchanged) threshold. Hence the condition remains violated at $t+1$, and by
hypothesis $J^*_{t+1}(x_{t+1})=0$. The continuation term vanishes and the probe
branch
$-t^{l_i}(1+\lambda h)+x_t s^{l_i}_i\big(T-\sum_{j=1}^{i-1}\mathbb{E}_{n_j}[\tau_j]-t_{i-1,i}\big)\le 0$,
so terminating is optimal and $J^*_t(x_t)=0$. When Eq.~\eqref{eq:threshold} holds,
the probe branch is at least
$-t^{l_i}(1+\lambda h)+x_t s^{l_i}_i\big(T-\sum_{j=1}^{i-1}\mathbb{E}_{n_j}[\tau_j]-t_{i-1,i}\big)\ge 0$
(using $J^*_{t+1}\ge 0$), so probing is optimal and $J^*_t(x_t)\ge 0$. The same
condition governs slot $t$.

By induction, Eq.~\eqref{eq:threshold} holds at every slot. Since
$x_t s^{l_i}_i$ is non-increasing (Lemma~\ref{lem:Monotone}) while the
threshold is constant across the slots of IoTD $i$, once the condition fails it
never holds again. Together with Lemma~\ref{lem:reducedpolicy}, the UAV probes
for a consecutive run and then terminates permanently.
\end{proof}


\begin{remark}
The aforementioned threshold holds for a non-trivial case wherein the random IoTD duty-cycle state is fixed during the entire tour, and the POMDP solver tracks the availability belief in a Bayesian manner. The other cases of independent random duty-cycling within the tour are not considered in this work.
\end{remark}

\subsubsection*{Computing $\mathbb{E}[n_i]$}
From Theorem~\ref{thm:threshold}, let $n_i$ be the total number of states until the threshold of Eq.~\eqref{eq:threshold} is met for the IoTD $i$. Then these $n_i$ states form transient states of the absorbing Markov chain induced by the underlying optimal threshold policy. Let $\mathbf{Q}_{n_i\times n_i}$ be the transient matrix of the absorbing Markov chain. The expected number of probes made by this absorbing Markov chain is given by:
\begin{equation}\label{eq:absexpsteps}
    \mathbb{E}[n_i]=\mathbf{e}^T_1(\mathbf{I}-\mathbf{Q})^{-1}\mathbf{1}
\end{equation}
where $\mathbf{1}$ is a column vector of ones of dimension $n_i\times 1$, and $e_1$ is a column vector of the same dimension with first element $1$ and all other elements $0$.

\par The Bellman equation of the stochastic shortest path problem at IoTD $i$ can be reconstructed using the fundamental matrix $(\mathbf{I}-\mathbf{Q})^{-1}$ and  expected stage rewards of the transient states $\mathbf{R}$ as follows:
\begin{equation}\label{eq:expvi}
    \mathbb{E}[V_i] = \mathbf{e}_1^T(\mathbf{I}-\mathbf{Q})^{-1}\mathbb{E}[\mathbf{R}]\mathbf{1}
\end{equation}
where $\mathbf{R}$ is the vector of expected rewards of the random variables $R_t$ as given in Eq.~\eqref{eq:stagereward} for $t=1$ to $n_i$. The algorithm for finding optimal policy solving (P3) is given in Algorithm~\ref{alg:pomdp}. Note that Algorithm~\ref{alg:pomdp} is the Lagrangian inner loop  for a given $\lambda$, and the dual problem for optimizing $\lambda$ is the outer-loop which is the bisection method as shown in Algorithm~\ref{alg:bisection}.

\begin{algorithm}[t]
\caption{Forward-Decomposition Approximation Algorithm}
\label{alg:pomdp}
\begin{algorithmic}[1]
\State \textbf{Input} Tour $\mathbf{z}$, hover positions $\mathbf{l} = (l_i)_{i=1}^{|z|}$,
         multiplier $\lambda$, availability vector $\mathbf{a} = (a_i)_{i=1}^{|z|}$
\State \textbf{Output}  Optimal probe vector $\mathbf{n}^*$, Lagrangian  $L(\mathbf{n}^*,\lambda; \mathbf{z},\mathbf{l})$
\State \textbf{Initialize} $temp \leftarrow T$, $\mathbf{n}\leftarrow \emptyset$, $L_{\text{acc}}\leftarrow 0$
\For{$i = 1$ \textbf{to} $|\mathbf{z}|$}
    \State $x_i \leftarrow a_i$,\; $n_i \leftarrow 0$
    \State $temp \leftarrow temp  - t_{i-1,i}$
           \Comment{Residual time after deducting flight time to IoTD $i$}
    \While{$x_i s^{l_i}_i \geq \dfrac{t^l_i(1 + \lambda h)}{temp}$}
           \Comment{Threshold test, Eq.~\eqref{eq:threshold}} 
        \State $n_i \leftarrow n_i + 1$
        \State $x_i \leftarrow \dfrac{x_i(1 - s^{l_i}_i)}{x_i(1 - s^{l_i}_i) + (1 - x_i)}$
               \Comment{Belief update, Eq.~\eqref{eq:bayesupdate}}
    \EndWhile
    \State Build transient matrix $\mathbf{Q}\in\mathbb{R}^{n_i\times n_i}$ of the absorbing chain
           induced by the $n_i$ consecutive probes
    \State $\mathbb{E}[n_i]=\mathbf{e}^T_1(\mathbf{I}-\mathbf{Q})^{-1}\mathbf{1}$ \Comment{From Eq.~\eqref{eq:absexpsteps}}
    \State $\mathbb{E}[V_i] = \mathbf{e}_1^T(\mathbf{I}-\mathbf{Q})^{-1}\mathbb{E}[\mathbf{R}]\,\mathbf{1}$ \Comment{From Eq.~\eqref{eq:stagereward}, Eq.~\eqref{eq:expvi}}
    \State $temp \leftarrow temp - \mathbb{E}[n_i]\, t^l_i$
           \Comment{Debit expected hover time to successor nodes}
    \State $L_{\text{acc}} \leftarrow L_{\text{acc}}+\mathbb{E}[V_i]$
           \Comment{Priced hover energy; flight energy accounted in CMA-ES layer}
    \State $\mathbf{n} \leftarrow \mathbf{n}\cup \{n_i\}$
\EndFor
\State $\mathbf{n}^* \leftarrow \mathbf{n}$
\State $L(\mathbf{n}^*,\lambda; \mathbf{z},\mathbf{l}) \leftarrow L_{\text{acc}}$
\State \Return $\mathbf{n}^*$, $L(\mathbf{n}^*,\lambda; \mathbf{z},\mathbf{l})$
\end{algorithmic}
\end{algorithm}

\par The following Proposition~\ref{prop:learn} relaxes the strong assumptions such as  device availability probability (as initial belief) to be known by showing the feasibility of learning the parameters of the optimal algorithm by using the Strong Law of Large Numbers (SLLN).
\begin{proposition}\label{prop:learn}
\textbf{Learning unknown parameters of IoTD-level POMDP}: There exists a feasible and optimal way of learning the parameters of IoTD-level POMDPs.
\end{proposition}
\begin{proof}
Consider the threshold structure developed in Thm.~\ref{thm:threshold}. From the \textit{L.H.S} of the threshold (as shown in Eq.~\eqref{eq:threshold}), it is obvious that learning the product of parameter terms $x_ts_i^{l_i}$ is sufficient to verify the threshold. By Lemma~\ref{lem:reducedpolicy}, the optimal policy lies in the reduced policy space with the structural form of $(1,1,\cdots,2)$. The absorbing Markov chain induced by such a reduced policy will visit each failed IoTD belief state ($x_t$) exactly once in every sample path, and most importantly in the same order. 
\par The responses at a given time $t$ (across episodes) follow a Bernoulli distribution with parameter $x_ts_i^{l_i}$. Hence, by SLLN, the sample average of responses converges to the ground truth parameter of $x_ts_i^{l_i}$ with probability $1$. 
\end{proof}

\subsection{Convergence and Complexity Analysis}\label{subsec:convergence}
{
\textit{Convergence.} (i) \emph{Probing:} Each probe consumes hover time and battery, so only finitely many probes are possible at an IoTD. After each failed probe, the belief decreases (Lemma~\ref{lem:Monotone}), while the threshold of Eq.~\eqref{eq:threshold} stays constant during the visit. By Theorem~\ref{thm:threshold}, the resulting threshold policy is optimal, and Eqs.~\eqref{eq:absexpsteps}--\eqref{eq:expvi} give its value exactly. (ii) \emph{CMA-ES and SA:} Both keep the best solution found so far, so the best objective never decreases. Both stop within a fixed budget of $g_{\max}$ generations and $R\,K_{\max}$ iterations for $R$ restarts, respectively. Since the tour problem is NP-hard~\cite{Orienteering_Problem}, global optimality cannot be guaranteed in polynomial time unless P\,=\,NP. SA reaches the global optimum with probability one only under logarithmic cooling~\cite{Hajek_SA}, whereas FLoRa uses geometric cooling for tractability. Fig.~\ref{fig:convergence} shows that, with one restart, SA reaches $95\%$ of its final VIP within about $860$ ($N=20$) and $1040$ ($N=40$) of $2000$ iterations, and CMA-ES within $24$--$26$ of $50$ generations.

\begin{proposition}\label{prop:complexity}
The offline planning cost of FLoRa is $O(R\,K_{\max}\,E\,N^3)$, and its online cost is $O(N\,\bar n)$, where $E$ is the number of evaluations per CMA-ES run and $\bar n$ is the last probe slot yielding a positive reward, i.e., the largest number of probes at an IoTD.
\end{proposition}
\begin{proof}[Proof sketch]
Let $|\mathbf{z}|\le N$ be the tour length, so the hover vector has dimension $2|\mathbf{z}|$. \emph{(i) Probing:} For one IoTD, Algorithm~\ref{alg:pomdp} finds the number of probes by the threshold test of Eq.~\eqref{eq:threshold} in $O(\bar n)$ steps. It computes Eqs.~\eqref{eq:absexpsteps}--\eqref{eq:expvi} from the fundamental matrix of the absorbing Markov chain, an $\bar n\times\bar n$ inversion in $O(\bar n^3)$. One evaluation of $F(\mathbf{l};\mathbf{z})$ applies this to the $|\mathbf{z}|$ IoTDs for each bisection step of Algorithm~\ref{alg:bisection}, whose number depends only on the tolerance $\epsilon$. One evaluation therefore costs $O(|\mathbf{z}|\,\bar n^3)$. \emph{(ii) Positioning:} One CMA-ES run performs $E$ evaluations. Sampling and the covariance update cost $O(|\mathbf{z}|^2)$ per candidate and the eigen-decomposition at most $O(|\mathbf{z}|^3)$~\cite{CMA_ES_Tutorial}, so one run costs $O\big(E\,(|\mathbf{z}|\,\bar n^3+|\mathbf{z}|^3)\big)$. \emph{(iii) Path planning:} SA invokes one CMA-ES run in each of its $R\,K_{\max}$ iterations, which gives $O\big(R\,K_{\max}\,E\,(|\mathbf{z}|\,\bar n^3+|\mathbf{z}|^3)\big)=O(R\,K_{\max}\,E\,N^3)$, since $\bar n$ does not depend on $N$. \emph{(iv) Online execution:} The UAV only updates its belief after each failed probe (Eq.~\eqref{eq:bayesupdate}) and compares it with the precomputed threshold. This takes $O(1)$ per probe, i.e., $O(|\mathbf{z}|\,\bar n)\subseteq O(N\,\bar n)$.
\end{proof}
\noindent Hence, FLoRa runs in time polynomial in $N$ and linear in the search budgets, whereas the exact problem grows factorially with $N$. This is consistent with the planning times in Table~\ref{tab:runtime}.
}

\begin{figure}[t]
    \centering
    \includegraphics[width=\columnwidth]{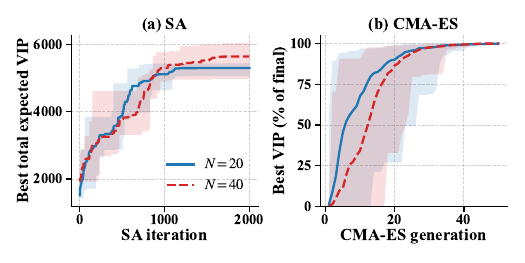}
      \caption{Convergence of (a) SA, best total expected VIP found so far, and (b) CMA-ES, best total expected VIP found so far as a percentage of the final value ($12$~km scenario, $b_{\max}=2568$~kJ, one restart, mean over $10$ seeds, shaded band: minimum--maximum).}
    \label{fig:convergence}
\end{figure}

\section{Well-Posedness of FLoRa Framework}\label{sec:florawellposedness}
\begin{enumerate}
\item \textit{Discontinuity regularization by CMA-ES}: The Gaussian based selection of location coordinates by CMA-ES~\cite{CMA_ES_Tutorial} acting as spatial filter  over the stair-case type LoRa channel  together with weight-based rank estimators  effectively help in approximately smoothing the (discontinuities of) abrupt SNR changes at the boundaries of SF levels. 
\item \textit{Common unit of the objective}:
All VIP terms are measured in time, so their sum in Eq.~\eqref{eq:main_problem} is a single objective that accounts for the coupling between IoTDs without weights.
    \item \textit{Imposing delay penalty}: The novelty of the proposed VIP metric is to impose a delay penalty, especially if the UAV abandons a node after failed probing attempts. Without the explicit delay tax,  POMDP solver would be chasing an inactive node by probing it multiple number of times with no feedback (zero reward/cost) thereby violating the third-well posedness condition (also known as continuous dependence condition) of Hadamard~\cite{hadamardwellposednessbook}. 
    \item \textit{Lagrangian relaxation}: The rigid and discontinuous hard-limit of the global shared energy constraint is regularized via a Lagrangian relaxation which transforms hard limit into soft limit by the energy constraint to be met in the expected sense.  
     \item  \textit{Structural regularization using bi-level energy optimization}: Primal decomposition approaches such as bi-level optimization~\cite{bilevelstructural} act as structural regularizers. The upper layers of SA/CMA-ES act as the leader (managing the hard flight energy), and the bottom layer of POMDP acts as the follower, which treats the flight energy as a constant given by the upper layer and optimizes the hovering energy. In this manner, the leader and the follower interact iteratively to solve the composite optimization problem. The interaction between POMDP and CMA-ES is through the scoring function (CMA-ES uses the underlying POMDP solver as its scoring function, as shown in \textit{Line 12} of Algorithm~\ref{alg:CMAES}).  
\end{enumerate}
Hence, the proposed FLoRa framework is, to the best of our knowledge, the first optimization model for combined UAV path planning, positioning, and probing which regularizes  various levels of rigidness and discontinuities  present in the problem structure.

\section{Results and Discussion}\label{sec:Results}
\subsection{Simulation Setup}
We evaluate FLoRa on two deployment scenarios that share the same relative
node layout but differ in LoRa coverage radius, spatial scale, and the
resulting path-loss regime, namely, a compact $3$~km scenario (dense, high-path-loss
urban field) and a wide-area $12$~km scenario (low-path-loss, near-free-space
rural field). The $12$~km node coordinates are a uniform $4\times$ scaling of
the $3$~km ones, and the success-probability profile is matched across
scenarios (Fig.~\ref{fig:snrsuccprb}), so the two use the same relative arrangement
and per-node availabilities while differing in scale, coverage radius, and
link budget. For each scenario we consider two densities, $N\!\in\!\{20,40\}$
IoTDs, with the $N=20$ set nested in $N=40$. Device availabilities $a_i$ are
drawn once from $\mathcal{U}[0,1]$ and held fixed across all methods and seeds
for a fair comparison. A single UAV with velocity
$v=50$~m/s collects data
under a freshness deadline $ T=1500$~seconds for the $3$~km scenario and $T=3600$~seconds
for the $12$~km scenario. The UAV battery budget $b_{\max}$ is swept in equal
steps, reported in kilojoules: $321$--$1605$~kJ for the $3$~km scenario and
$642$--$3852$~kJ for the $12$~km scenario. Results are averaged over at least $10$ independent runs per point for the stochastic methods (FLoRa, SSA, ACO, PPO, Safe-TD3, and EAMO) while GU and GUHO are deterministic. Error bars in all figures denote $95\%$ confidence intervals. The remaining hyper-parameters are listed in Table~\ref{tab:sim_params}.

\begin{table}[ht]
\centering
\caption{Simulation parameters}
\label{tab:sim_params}
\renewcommand{\arraystretch}{1}
\begin{tabular}{|l|l|c|c|}
\hline
\textbf{Category} & \textbf{Parameter} & \textbf{3 km} & \textbf{12 km} \\
\hline
\multirow{6}{*}{Network}
  & Deployment area               & $78.75 \times 55$ & $315 \times 220$ \\
  &                             & km$^2$ & km$^2$ \\
  & Number of IoTDs ($N$)       & \multicolumn{2}{c|}{20, 40} \\
  & IoTD radius ($d_{\max}$)         & 3 km & 12 km \\
  & Device availability ($a_i$) & \multicolumn{2}{c|}{$[0,\,1]$} \\
\hline
\multirow{9}{*}{Channel / Radio}
  & Bandwidth ($BW$)            & \multicolumn{2}{c|}{125 kHz} \\
  & Carrier frequency ($f_c$)   & \multicolumn{2}{c|}{869.4 MHz} \\
  & Noise power density ($N_0$) & \multicolumn{2}{c|}{$-174$ dBm/Hz} \\
  & Noise figure ($NF$)         & \multicolumn{2}{c|}{6 dB} \\
  & Transmit power ($p_t$)      & \multicolumn{2}{c|}{19 dBm} \\
  & Path-loss exponent ($\gamma$) & \multicolumn{2}{c|}{2.96} \\
  & Excess path loss ($L_{ex}$) & 17.82 dB & 0 dB \\
  & Payload size                & \multicolumn{2}{c|}{25 bytes} \\
\hline
\multirow{4}{*}{UAV}
  & Velocity ($v$)              & \multicolumn{2}{c|}{50 m/s} \\
  & Flight power ($f$)          & \multicolumn{2}{c|}{321 W} \\
  & Hover power ($h$)           & \multicolumn{2}{c|}{168 W} \\
  & Battery budget ($b_{\max}$) & 321--1605 kJ & 642--3852 kJ \\
\hline
VIP Metric
  & Deadline ($T$)              & 1500 s & 3600 s \\
\hline
\multirow{5}{*}{SA}
  & Cooling rate ($\alpha$)     & \multicolumn{2}{c|}{0.995} \\
  & Initial temperature ($\Upsilon_0$) & \multicolumn{2}{c|}{5000} \\
  & Stopping criterion          & \multicolumn{2}{c|}{$\Upsilon < 0.01$} \\
  & Independent restarts ($R$)  & 5 & 1 \\
  & $K_{\tt max}$              &  \multicolumn{2}{c|}{2000 per restart}\\
\hline
\multirow{2}{*}{CMA-ES}
  & Initial step size ($\sigma_0$) & 1.25 & 5.0 \\
  & Population size ($\Gamma$)   & \multicolumn{2}{c|}{adaptive to $\mathbf{|z|}$} \\
\hline
$\lambda$ Search
  & Lagrange multiplier         & \multicolumn{2}{c|}{$[0,\,1/h]$ via bisection} \\
\hline
\end{tabular}
\end{table}
\subsection{Baseline Approaches}\label{sec:baselines}
The proposed FLoRa framework is compared against eight baseline methods, of which the first five are generic methods adapted to PTCOP and the last three are recent methods from the literature (Table~\ref{tab:related_work_comparison}).

\begin{enumerate}
  \item \emph{Simple\_SA (SSA)} uses the same SA-based global tour planning as
  FLoRa but hovers directly above each IoTD, omitting CMA-ES hover
  optimization. The SSA--FLoRa gap therefore isolates the contribution of
  close-enough positioning.
  \item \emph{Greedy\_Util (GU)} greedily selects the next IoTD by
  utility-to-travel-time ratio and hovers directly above each node.
  \item \emph{Greedy\_Util\_Hover\_Optimization (GUHO)} augments GU with
  CMA-ES hover optimization, matching FLoRa's positioning but retaining greedy
  (rather than SA-based) path planning. The GUHO--FLoRa gap isolates the contribution of global path
  optimization.
  \item \emph{Ant Colony Optimization (ACO)}~\cite{ACO} is an alternative metaheuristic
  tour planner that, like SSA, hovers above each node.
    \item \emph{Proximal Policy Optimization (PPO)}~\cite{PPO} is a deep-reinforcement-learning baseline where the UAV agent jointly learns node selection and hover positioning.
  To ensure a strong baseline, PPO is trained independently from scratch on each fixed network topology, rather than as a single generalized agent, which would underperform on any individual layout. This gives PPO its best per-instance performance, so that FLoRa is not compared against a generalization-averaged result.
  \item \emph{Safe-TD3}~\cite{SafeTD3} is a recent constrained DRL method with a hard energy budget. Its algorithm and hyper-parameters are kept as in~\cite{SafeTD3}, and the settings not given there are a learning rate of $10^{-4}$, a batch size of $256$, and two hidden layers of $256$ units. Since TD3 requires continuous actions, the actor outputs node scores (the best feasible node is visited) and a hover offset. The reward is the VIP, and the cost is the energy per step normalized by the budget. Like PPO, it is retrained on each configuration for $16000$ episodes.
  \item \emph{EAMO}~\cite{LoRa_UAV_AoI} minimizes a weighted sum of energy and AoI with a genetic tour, optimized hover positions, and a fixed number of probes per IoTD. Since it visits all IoTDs, it is adapted to the budget by removing IoTDs and re-planning, and its weights, probing target, and hover range are tuned for each battery. EAMO+POMDP executes the same plans with FLoRa's probing policy.
  \item \emph{Li et al.}~\cite{OP_DC} greedily select hover points, each serving several IoTDs, by the collected data per unit of energy. It is run with our energy and LoRa models and FLoRa's probing policy (Li+POMDP), with its coverage radius and hover time tuned for each battery.
\end{enumerate}
For a fair comparison, PPO, Safe-TD3, EAMO+POMDP, and Li+POMDP use the same environment and POMDP probing policy as FLoRa (Section~\ref{sec:POMDP}), so that they differ from FLoRa only in planning.

The performance of FLoRa and the baseline methods is evaluated using 
three metrics: 
\begin{itemize}
    \item \textbf{Total Expected VIP~}($\sum_{i=1}^{|\mathbf{z}|}\mathbb{E}[V_i]$) is the primary metric, capturing the end-to-end system performance
by measuring whether the UAV collects fresh, valuable data within the deadline $T$. It reflects the joint contribution of all three
FLoRa components: path planning, hover optimization, and adaptive
probing.
    \item \textbf{Node Coverage $|\mathbf{z}^*|$} measures the number of IoTDs
visited by the UAV, evaluating the spatial reach of the tour and
primarily reflecting the effectiveness of SA-based path planning
under the battery constraint.
    \item \textbf{Expected successful collections} from Section~\ref{sec:POMDP} measures the expected collection success across all probe attempts in the tour,
validating the effectiveness of the POMDP-based adaptive probing
strategy under unavailability uncertainty.

\end{itemize}

\begin{figure*}[t]
    \centering
    \includegraphics[width=\textwidth]{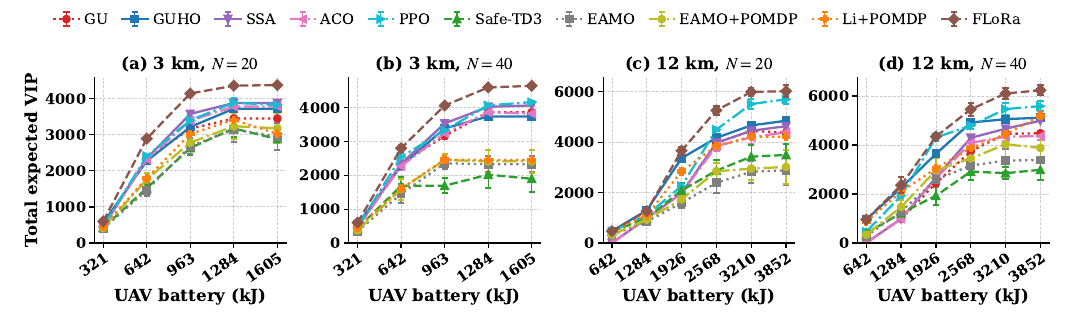}
    \caption{Total expected VIP vs UAV battery ($b_{\max}$): (a), (b) $3$~km and (c), (d) $12$~km scenario, for $N=20$ and $N=40$ IoTDs.}
    \label{fig:Expected_VIP}
\end{figure*}

\subsection{Total Expected VIP Performance}
Fig.~\ref{fig:Expected_VIP} reports the total expected VIP
against battery budget for the $3$~km and $12$~km scenarios. 
FLoRa achieves the highest total expected VIP in every scenario.

\textbf{Close-enough positioning yields feasibility at low energy.}
The clearest effect appears at the lowest $12$~km budget ($642$~kJ). Here
SSA, GU, and ACO collect zero VIP: because they hover directly above a node, the
round trip to even the nearest IoTD exceeds the available energy, so no data
is collected. FLoRa, by hovering at the coverage boundary of the
nearest node (via CMA-ES), still reaches one IoTD and collects
$469$ ($N=20$) and $955$ ($N=40$) units of VIP, a feasibility margin that
the above-node baselines cannot achieve at any seed. This is a direct,
physical demonstration of the close-enough paradigm: trading a marginal
reduction in packet success probability $s_i^{l_i}$ for a large reduction in
flight cost converts an infeasible mission into a productive one.

\textbf{Ablation: attributing the gains.} Averaged over all battery levels,
the FLoRa vs SSA comparison (isolating CMA-ES hovering) yields VIP gains of
$15.9\%$ ($3$~km) and $38.9\%$ ($12$~km), while the FLoRa vs GUHO comparison
(isolating SA-based global path planning, with CMA-ES hovering held common)
yields gains of $21.7\%$ ($3$~km) and $18.3\%$ ($12$~km). Together, these two
ablations show that both close-enough positioning and global path planning contribute to the gain. When the genetic search of EAMO or the greedy selection of Li et al. maximizes the expected VIP with the same probing, FLoRa remains $7$--$18\%$ and $21$--$30\%$ higher, respectively.

Against the DRL baseline PPO, FLoRa improves total expected VIP by $15.2\%$ on
average. This advantage reflects FLoRa's hierarchical
decomposition, wherein SA searches the discrete tour space directly and
CMA-ES exploits the non-differentiable SF-transition ridges (Fig.~\ref{fig:snrsuccprb})
via covariance adaptation, whereas PPO must learn both decisions jointly from
a single policy network under a discontinuous, non-differentiable reward
surface.

\textbf{Comparison with recent methods.} FLoRa outperforms Safe-TD3, EAMO, EAMO+POMDP, and Li+POMDP by $54$--$116\%$, $54$--$109\%$, $44$--$92\%$, and $30$--$77\%$ in total expected VIP (Tables~\ref{tab:impr3km} and~\ref{tab:impr12km}). Unlike the generic baselines, which score each tour with the exact expected VIP, Safe-TD3 must learn this value from sampled rewards, and EAMO and Li et al. optimize energy and AoI, and data per unit of energy, respectively, without accounting for freshness under random availability.

\begin{figure*}[t]
    \centering
\includegraphics[width=0.85\textwidth]{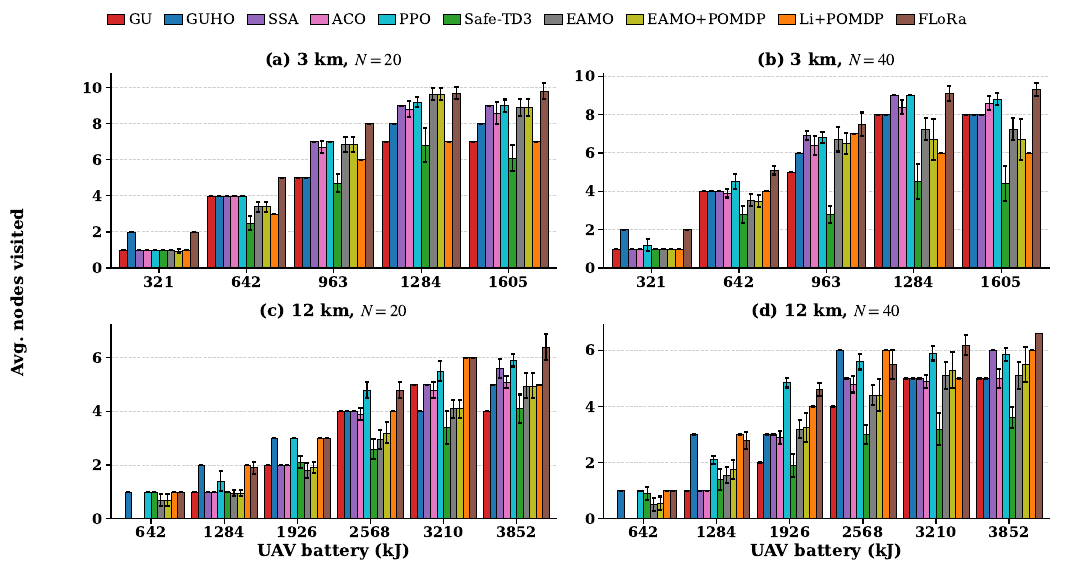}
    \caption{Node Coverage vs UAV battery ($b_{\max}$): (a), (b) $3$~km and (c), (d) $12$~km scenario, for $N=20$ and $N=40$ IoTDs.}
    \label{fig:Coverage}
\end{figure*}

\subsection{Node Coverage Analysis}
Fig.~\ref{fig:Coverage} reports the number of IoTDs
visited. Node Coverage rises with battery for all methods. FLoRa attains the
highest Node Coverage at the maximum budget, up to $9.8$ nodes ($3$~km, $N=20$) and
$6.6$ nodes ($12$~km, $N=40$) because SA-based global planning selects the
high-value subset rather than committing to locally greedy choices. CMA-ES
contributes indirectly: by trimming per-node flight cost relative to the
above-node baselines, it preserves energy for additional visits, which is why
FLoRa consistently visits more nodes than SSA at equal battery.

At low and medium $12$~km budgets, GUHO, PPO, and Li+POMDP occasionally match or exceed FLoRa in Node Coverage despite trailing in total VIP, since they favor cheaply reachable nodes without accounting for the freshness cost imposed on later nodes, whereas SA optimizes the VIP objective (Eq.~\eqref{eq:vip}) directly.

Notably, at $12$~km, PPO's Node Coverage deficit is the smallest of all baselines ($+6.2\%$ in Table~\ref{tab:impr12km}), indicating that PPO's
learned policy is comparatively competent at how many nodes to visit,
even though its VIP and success-probability gaps show it is far
less effective at which nodes and how to hover once a node is
chosen.

Averaged over
all battery levels, FLoRa improves Node Coverage over GUHO by $22.8\%$ ($3$~km) and
$18.8\%$ ($12$~km), and over SSA by $14.6\%$ ($3$~km) and $27.0\%$ ($12$~km).
GU and ACO trail throughout, reflecting the limits of greedy/above-node tour
construction under a tight energy budget.

\begin{figure*}[t]
    \centering
    \includegraphics[width= 0.85\textwidth]{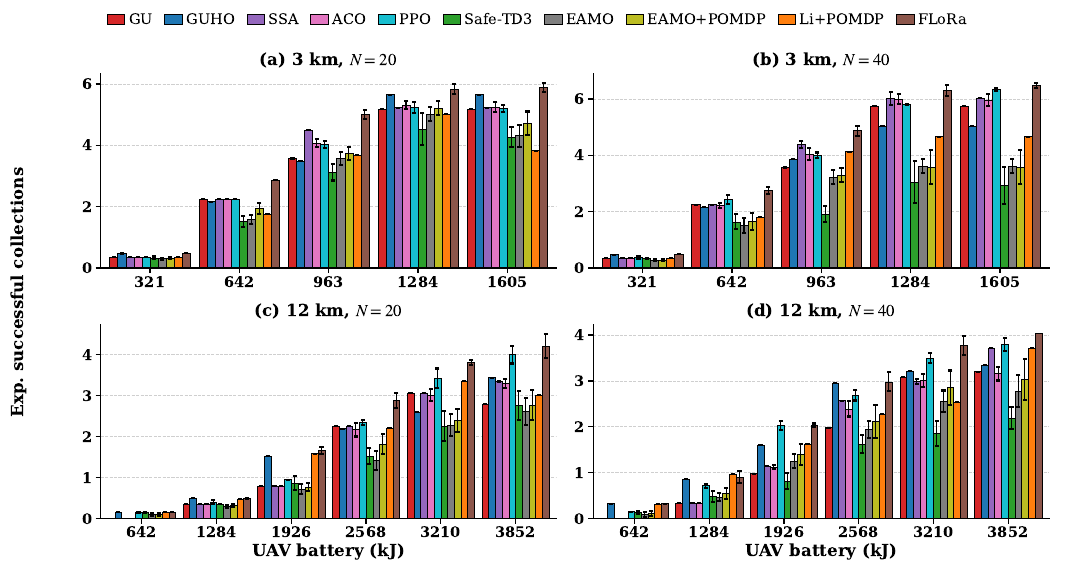}
    \caption{Expected successful collections vs UAV battery ($b_{\max}$): (a), (b) $3$~km and (c), (d) $12$~km scenario, for $N=20$ and $N=40$ IoTDs.}
    \label{fig:Success}
\end{figure*}

\subsection{Expected Successful Collections Analysis}
Fig.~\ref{fig:Success} shows that FLoRa accumulates the highest total success across all budgets, and this comes from how the three layers work together. SA searches over tours, deciding which nodes to visit and in what order, and its ordering naturally ends up front-loading high-availability nodes, visiting them while the residual freshness is still large. CMA-ES then searches for the best hover position at each of these nodes. Its job is not simply to hover near the coverage boundary and save flight cost for the sake of coverage and that would come at the price of a lower $s_i^{l_i}$ everywhere. Instead, CMA-ES looks for a position that keeps $s_i^{l_i}$ reasonably high while still saving enough battery, so the reach it buys does not come at the cost of reliability. With a high $s_i^{l_i}$ and a large $T-\sum_{j=1}^{i-1}\mathbb{E}_{n_j}[\tau_j]$ already in place, the POMDP's threshold policy Eq.~\eqref{eq:threshold} stays low, which lets it probe each node more times, and this is what actually drives up the cumulative success count. Averaged across all battery levels, FLoRa improves total success over GUHO by $20.7\%$ ($3$~km) and $20.6\%$ ($12$~km), and over PPO by $14.0\%$ ($3$~km) and $13.1\%$ ($12$~km).

\par Tables~\ref{tab:impr3km} and~\ref{tab:impr12km} summarize the performance of the proposed FLoRa compared with other baseline approaches across the aforementioned performance metrics. The table shows the average performance improvement across different UAV budgets, and the last column for each metric shows the average performance across different configurations and energy budgets. The improvement is computed on each metric summed over the battery levels at which the baseline collects VIP.

\begin{table}[ht]
\centering
\caption{Percentage improvement of FLoRa over baselines, averaged over all
UAV battery levels ($3$~km scenario).}
\label{tab:impr3km}
\setlength{\tabcolsep}{2.5pt}
\footnotesize

\resizebox{\columnwidth}{!}{%
\begin{tabular}{l|ccc|ccc|ccc}
\hline
 & \multicolumn{3}{c|}{Total Expected VIP} & \multicolumn{3}{c|}{Coverage} & \multicolumn{3}{c}{Successful collections}\\
Baseline & $N{=}20$ & $N{=}40$ & Avg & $N{=}20$ & $N{=}40$ & Avg & $N{=}20$ & $N{=}40$ & Avg\\
\hline
GUHO & +21.4 & +22.0 & \textbf{+21.7} & +27.8 & +17.9 & \textbf{+22.8} & +15.2 & +26.1 & \textbf{+20.7} \\
ACO & +19.2 & +22.3 & \textbf{+20.8} & +18.6 & +16.6 & \textbf{+17.6} & +16.6 & +12.7 & \textbf{+14.6} \\
PPO & +17.1 & +14.6 & \textbf{+15.8} & +14.2 & +8.9 & \textbf{+11.6} & +17.7 & +10.4 & \textbf{+14.0} \\
SSA & +15.8 & +15.9 & \textbf{+15.9} & +15.0 & +14.2 & \textbf{+14.6} & +14.3 & +10.0 & \textbf{+12.1} \\
GU & +27.3 & +22.1 & \textbf{+24.7} & +43.8 & +26.9 & \textbf{+35.3} & +21.6 & +18.6 & \textbf{+20.1} \\
\hline
Safe-TD3 & +54.0 & +116.1 & \textbf{+85.1} & +63.5 & +112.9 & \textbf{+88.2} & +45.8 & +111.7 & \textbf{+78.8} \\
EAMO & +53.9 & +89.7 & \textbf{+71.8} & +15.8 & +28.2 & \textbf{+22.0} & +35.8 & +70.9 & \textbf{+53.3} \\
EAMO+POMDP & +44.0 & +80.8 & \textbf{+62.4} & +16.0 & +35.2 & \textbf{+25.6} & +26.2 & +68.6 & \textbf{+47.4} \\
\hline
Li+POMDP & +39.8 & +77.4 & \textbf{+58.6} & +43.8 & +37.5 & \textbf{+40.6} & +37.2 & +34.0 & \textbf{+35.6} \\
\hline
\end{tabular}}
\end{table}

\begin{table}[ht]
\centering
\caption{Percentage improvement of FLoRa over baselines, averaged over all
UAV battery levels ($12$~km scenario).}
\label{tab:impr12km}
\setlength{\tabcolsep}{2.5pt}
\footnotesize

\resizebox{\columnwidth}{!}{%
\begin{tabular}{l|ccc|ccc|ccc}
\hline
 & \multicolumn{3}{c|}{Total Expected VIP} & \multicolumn{3}{c|}{Coverage} & \multicolumn{3}{c}{Successful collections}\\
Baseline & $N{=}20$ & $N{=}40$ & Avg & $N{=}20$ & $N{=}40$ & Avg & $N{=}20$ & $N{=}40$ & Avg\\
\hline
GUHO & +20.5 & +16.1 & \textbf{+18.3} & +21.6 & +16.1 & \textbf{+18.8} & +26.9 & +14.3 & \textbf{+20.6} \\
ACO & +44.3 & +49.2 & \textbf{+46.8} & +31.5 & +38.2 & \textbf{+34.9} & +35.4 & +37.0 & \textbf{+36.2} \\
PPO & +16.1 & +12.9 & \textbf{+14.5} & +6.9 & +5.5 & \textbf{+6.2} & +17.1 & +9.1 & \textbf{+13.1} \\
SSA & +38.7 & +39.1 & \textbf{+38.9} & +25.6 & +28.5 & \textbf{+27.0} & +33.0 & +27.7 & \textbf{+30.3} \\
GU & +41.1 & +52.1 & \textbf{+46.6} & +38.1 & +51.2 & \textbf{+44.7} & +41.0 & +43.3 & \textbf{+42.1} \\
\hline
Safe-TD3 & +70.5 & +107.0 & \textbf{+88.7} & +62.7 & +90.7 & \textbf{+76.7} & +66.9 & +98.0 & \textbf{+82.4} \\
EAMO & +109.2 & +81.1 & \textbf{+95.1} & +49.5 & +34.5 & \textbf{+42.0} & +77.9 & +54.6 & \textbf{+66.2} \\
EAMO+POMDP & +92.2 & +57.8 & \textbf{+75.0} & +46.2 & +28.7 & \textbf{+37.4} & +61.4 & +39.6 & \textbf{+50.5} \\
\hline
Li+POMDP & +35.9 & +30.5 & \textbf{+33.2} & +10.0 & +6.8 & \textbf{+8.4} & +22.3 & +23.0 & \textbf{+22.6} \\
\hline
\end{tabular}}\\[2pt]\parbox{\columnwidth}{\scriptsize At $642$~kJ, SSA, ACO, and GU collect no VIP, so this budget is excluded from their averages.}
\end{table}

Table~\ref{tab:impr_cat_final} further aggregates these results into five baseline categories: Greedy (GU, GUHO), Metaheuristic (SSA, ACO), DRL (PPO, Safe-TD3), EAMO (EAMO, EAMO+POMDP), and Li et al. (Li+POMDP).
\begin{table}[ht]
\centering
\caption{Percentage improvement of FLoRa over baseline categories.}
\label{tab:impr_cat_final}
\setlength{\tabcolsep}{3pt}
\footnotesize

\resizebox{\columnwidth}{!}{%
\begin{tabular}{l|ccc|ccc|ccc}
\hline
 & \multicolumn{3}{c|}{3~km} & \multicolumn{3}{c|}{12~km} & \multicolumn{3}{c}{Overall}\\
Category & VIP & Cov. & Succ. & VIP & Cov. & Succ. & VIP & Cov. & Succ. \\
\hline
Metaheuristic & +18.3 & +16.1 & +13.4 & +42.9 & +30.9 & +33.3 & \textbf{+30.6} & \textbf{+23.5} & \textbf{+23.3} \\
Greedy & +23.2 & +29.1 & +20.4 & +32.4 & +31.7 & +31.4 & \textbf{+27.8} & \textbf{+30.4} & \textbf{+25.9} \\
DRL & +50.5 & +49.9 & +46.4 & +51.6 & +41.5 & +47.8 & \textbf{+51.0} & \textbf{+45.7} & \textbf{+47.1} \\
EAMO & +67.1 & +23.8 & +50.4 & +85.1 & +39.7 & +58.4 & \textbf{+76.1} & \textbf{+31.8} & \textbf{+54.4} \\
Li et al. & +58.6 & +40.6 & +35.6 & +33.2 & +8.4 & +22.6 & \textbf{+45.9} & \textbf{+24.5} & \textbf{+29.1} \\
\hline
\end{tabular}}
\end{table}

\subsection{Sensitivity Analysis}

\begin{figure*}[t]
  \centering
  \includegraphics[width=\textwidth]{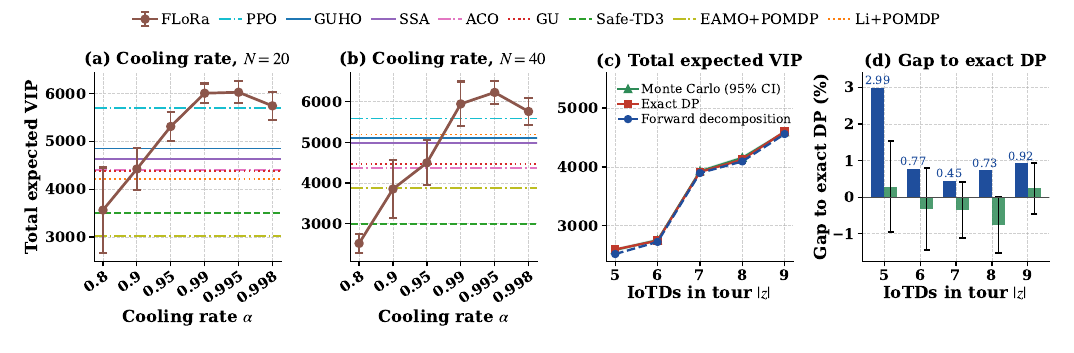}
  \caption{(a), (b) SA sensitivity w.r.t.\ the cooling rate $\alpha$ ($12$~km scenario, $b_{\max}=3852$~kJ, horizontal lines: baselines at the same budget), (c) total expected VIP of exact DP, forward decomposition, and Monte Carlo (95\% CI) versus tour length, and (d) gap to exact DP (Eq.~\eqref{eq:fwd_decomp_gap}).}
  \label{fig: alpha_sensitivity}
\end{figure*}
Fig.~\ref{fig: alpha_sensitivity}(a), (b) reports the total expected VIP as
a function of the SA cooling rate $\alpha$ for $N=20$ and $N=40$, with
error bars denoting $95\%$ confidence intervals over independent runs.
The behavior follows from the cooling schedule $\Upsilon_{\mathit{itr}}=\Upsilon_0\alpha^{\mathit{itr}}$: SA stops after $\mathit{itr}_{\text{stop}}=\lceil\ln(\Upsilon_{\min}/\Upsilon_0)/\ln\alpha\rceil$ iterations or at $K_{\max}$. With $\Upsilon_0=5000$ and $\Upsilon_{\min}=0.01$, $\mathit{itr}_{\text{stop}}=59$, $125$, and $256$ for $\alpha=0.8$, $0.9$, and $0.95$, i.e., SA uses at most $13\%$ of its budget, freezes in an early local optimum, and yields low VIP. For $\alpha=0.99$, $\mathit{itr}_{\text{stop}}=1306$, and for $\alpha=0.995$ the final temperature at $K_{\max}$ is $0.22$, so SA explores and then refines. For $\alpha \in [0.99, 0.995]$, the VIP reaches a
stable maximum for both network sizes, indicating a broad
well-conditioned operating region rather than a single sharp optimum.
The close agreement between the $N=20$ and $N=40$ curves further confirms
that the optimal cooling range is insensitive to network density. For $\alpha=0.998$, the temperature at $K_{\max}$ is still $91.2$, so worse tours are still accepted when the budget ends and the VIP declines modestly. We
therefore adopt $\alpha = 0.995$, a conventional cooling rate, which lies
within the stable plateau and exhibits the lowest variance. The horizontal lines in Fig.~\ref{fig: alpha_sensitivity}(a), (b) show that, in the stable range $\alpha\in[0.99,0.995]$, FLoRa outperforms all baselines, whereas for $\alpha\le0.95$ it falls below PPO. The gain of FLoRa therefore relies on adequate exploration by SA.

\subsection{Sensitivity Analysis of Forward Decomposition Approximation}
\label{subsec:Fwd_decomposition}

The forward decomposition method designed in Section~\ref{sec:POMDP} decouples the
chain of dependent POMDPs so that each IoTD's POMDP can be independently solved. The aforementioned approach is tractable as it decouples the  statistical coupling between the linear (Markov-dependent) chain of IoTDs. To evaluate the loss in accuracy due to the aforementioned approximation, we performed a sensitivity study by considering fixed tours of length $|\mathbf{z}| \in \{5,6,7,8,9\}$ and comparing against two independent references. As the ground truth (coupled POMDP) is analytically complicated (due to the curse of dimensionality), we restricted the study to a few nodes. As shown in Fig.~\ref{fig: alpha_sensitivity}(d), the approximation loss is at most $2.99\%$ ($|\mathbf{z}|=5$) and $1.17\%$ on average.

\par We also compare the performance of the node-level forward decoupling method with the following 
\textbf{Monte Carlo} estimate:  we simulate FLoRa's own probing
policy over $N_{\text{MC}}$ independent rollouts of the  duty-cycling
process, giving an unbiased empirical estimate of the VIP that policy actually
earns, reported as a mean with a 95\% confidence interval. The exact DP is
thus an upper bound produced by the optimal policy, while the forward decoupling method and Monte
Carlo values describe the same policy that FLoRa actually deploys
analytically and by simulation, respectively.

Fig.~\ref{fig: alpha_sensitivity}(c) plots the total expected VIP for all three methods against tour length, and Fig.~\ref{fig: alpha_sensitivity}(d) shows the gap of the forward decomposition to the exact DP (and likewise of the Monte Carlo mean):
\begin{equation}
    \text{Gap}(\%) = \frac{R_{\mathrm{Exact}} - R_{\mathrm{Fwd}}}{R_{\mathrm{Exact}}} \times 100.
    \label{eq:fwd_decomp_gap}
\end{equation}

The mean of the Monte Carlo method stays close to the exact DP value at every tour length, with the exact value inside the $95\%$ confidence interval. The forward decomposition is therefore slightly conservative: it underestimates the value that the deployed policy earns by at most $2.99\%$. Since the simulated policy is FLoRa's own policy rather than the exact one, this agreement is the substantive
result: the policy FLoRa actually runs recovers essentially all of the reward
available to the optimal chain solver.

\subsection{Runtime Analysis}

\begin{table}[ht]
\centering
\caption{Offline planning time comparison across methods (mean, single problem instance, $N{=}40$, 12\,km scenario).}
\label{tab:runtime}
\begin{tabular}{lccc}
\toprule
\textbf{Method} & \textbf{Path Planning} & \textbf{Hover Opt.} & \textbf{Mean Time} \\
\midrule
GU    & Greedy & None (hover-on-node) & $\sim$2\,s \\
SSA   & SA     & None (hover-on-node) & $\sim$6.2\,s \\
GUHO  & Greedy & CMA-ES               & $\sim$112\,s \\
PPO   & Learned (DRL) & Learned (DRL) & $\sim$308\,s (5.1\,min) \\
\textbf{FLoRa} & SA & CMA-ES (+ nested POMDP) & $\sim$724\,s (12.1\,min) \\
\bottomrule
\end{tabular}
\end{table}
Table~\ref{tab:runtime} reports the mean offline planning time per method on a representative instance ($N{=}40$, 12\,km), averaged over 20 seeds on a MacBook Air M1 (8GB RAM, no GPU). All methods plan offline; online execution merely follows the precomputed tour, hover positions, and closed-form threshold policy (Theorem~\ref{thm:threshold}), imposing no real-time burden. The hierarchy reflects architecture: GU and SSA skip hover optimization and finish in seconds; adding CMA-ES hover optimization (GUHO) costs $\sim$50$\times$ more, as each step re-optimizes positions via derivative-free search; and FLoRa's further $\sim$6.5$\times$ over GUHO arises because SA invokes the CMA-ES/POMDP evaluation at every perturbation, making cost multiplicative in (SA iterations)\,$\times$\,(CMA-ES cost) rather than additive. PPO is trained per configuration, like all methods here, so its time is the same per-instance cost via a different paradigm; under this like-for-like comparison FLoRa ($\sim$12.1\,min) is $\sim$2.3$\times$ PPO ($\sim$5.1\,min).

\section{Conclusion}\label{sec:conclusion}
In this paper, we addressed the challenging problem of UAV-assisted data  collection over duty-cycling LoRa IoT networks under energy  and deadline constraints. We proposed FLoRa, a novel Flight-assisted data collection  framework over LoRa nodes, formulated as a Probabilistic, Time-dependent,  and Close-enough Orienteering Problem (PTCOP). FLoRa is a hierarchical framework involving $i)$ SA algorithm for path planning, $ii)$ CMA-ES for positioning, and $iii)$ POMDP for probing. The network-level POMDP is decomposed into IoTD-level POMDPs using Lagrangian relaxation and forward decomposition approximation as surrogate methods.

\par Extensive simulation-based evaluation demonstrated that FLoRa consistently outperforms all baseline categories in total expected VIP, with average gains of $30.6\%$ over metaheuristic baselines, $27.8\%$ over greedy baselines, $15.2\%$ over the strongest deep reinforcement learning baseline (PPO), $86.9\%$ over the recent Safe-TD3, $76.1\%$ over the recent EAMO (with and without POMDP probing), and $45.9\%$ over the energy-constrained planner of Li et al. This advantage carries over to the remaining metrics as well: node coverage improves by $23.5\%$, $30.4\%$, and $8.9\%$, and expected successful collections by $23.3\%$, $25.9\%$, and $13.6\%$, over the metaheuristic, greedy, and PPO baselines, respectively. These results confirm the benefit of jointly optimizing path planning, hover positioning, and adaptive probing for fresh data collection in energy-constrained UAV-LoRa networks. In the future, we plan to extend FLoRa to multi-UAV settings involving a swarm of UAVs performing data collection over ground LoRa IoTDs.

\bibliographystyle{ieeetr}
\tiny \bibliography{Transactions-Bibliography/references}
\normalsize

\clearpage
\appendices
\section{Proof of Lemma~\ref{lem:reducedpolicy}}\label{app:lemma1}
{
\begin{proof}
We use an exchange argument. Consider any node $i$ and any admissible policy $\pi$ that idles before probing again, i.e.\ contains the sub-pattern $(\dots,1,0,1,\dots)$: a probe, followed by at least one idle slot, followed by a further probe. Let $\pi'$ be the policy obtained by deleting the intervening idle slot(s), yielding $(\dots,1,1,\dots)$.

By the idle belief update in Eq.~\eqref{eq:idletransition}, an idle action leaves the belief unchanged, $x_{t+1}=x_t$. Hence under both $\pi$ and $\pi'$ the probe that follows the idle slot(s) acts on the same belief state, so both policies (i) collect data with identical probability and (ii) obtain the identical freshness reward on success. The two policies therefore yield the same expected VIP reward.

They differ only in cost. Each idle slot in $\pi$ requires the UAV to continue hovering at node $i$, consuming hover time $t_i^l$ and priced energy $\lambda h\,t_i^l$ per slot, and reducing the residual freshness available to all subsequent nodes by $t_i^l$ per slot, without altering the belief or advancing the chance of success. Consequently
\[
   L_i(\pi';z,l,\lambda) \;\ge\; L_i(\pi;z,l,\lambda),
\]
with strict inequality whenever $\pi$ contains at least one idle slot. Thus any policy containing an idle-then-probe pattern is weakly dominated by the policy that removes the idle slot.

Applying this exchange repeatedly eliminates every idle slot that precedes a later probe, leaving a run of consecutive probes with no interleaved idling. Since idle is dominated, the only non-probe action that survives on the optimal path is the terminal terminate action; the resulting policy has the form $(1,\dots,1,2)$, i.e.\ the UAV probes for $n_i=\sum_t n_t^{\pi}$ consecutive attempts and then terminates. Hence the reduced policy space $\tilde{\Pi}$ of consecutive-probe-then-terminate policies contains an optimal policy, and the per-slot optimal decision reduces to probe versus terminate.
\end{proof}
}

\section{Proof of Lemma~\ref{lem:Monotone}}\label{app:lemma2}
{
\begin{proof}
We use a `proof by contradiction' approach. Assume the following contradicting statement $x_{t+1} > x_t$ is true. This is trivially not true for the `terminate' action as per Eq.~\eqref{eq:terminatetransition}. For $x_t>0$, by applying the belief update shown in Eq.~\eqref{eq:bayesupdate} to the aforementioned contradicting statement, we get the following:
\begin{equation}\label{eq:contradictupdate}
\frac{x_t(1-s_i^{l_i})}{x_t(1-s_i^{l_i}) + (1-x_t)} > x_t
\end{equation}
Eq.~\eqref{eq:contradictupdate} can be rewritten as,
\begin{equation}
x_t(1-s_i^{l_i}) > x_t\left(x_t(1-s_i^{l_i}) + (1-x_t)\right)
\end{equation}
Dividing both sides of the aforementioned equation by $x_t$, we get the following:
\begin{equation}
1-s_i^{l_i} > x_t(1-s_i^{l_i}) + (1-x_t)
\end{equation}
Rearranging the above equation yields the following:
\begin{equation}
1-s_i^{l_i} > 1 - x_ts_i^{l_i}
\nonumber\end{equation}
\begin{equation}
x_ts_i^{l_i} > s_i^{l_i}
\end{equation}
\begin{equation}
x_t > 1
\end{equation}
By axioms of probability, the belief state probability $x_t>1$ is a contradiction. Hence $x_{t+1} \le x_t$ holds true.
\end{proof}
}

\section{Comparison of Freshness Metrics}\label{app:freshness}
The data collected at each IoTD is considered to carry a time-varying value $V_T(t)$ subject to a common deadline $T$. After the deadline, the value of the data is completely lost. Such a deadline-based metric (for instance, VoI~\cite{VoI_with_deterministic_deadline}) that decays over time and becomes void after the deadline is necessary for disaster response and other time-critical applications.

{
\begin{itemize}
  \item \emph{AoI}~\cite{Kaul_AoI}: $\mathrm{AoI}(t)=t-t_g(t)$, where $t_g(t)$ is the generation time of the latest received update. It grows without bound and has no deadline.
  \item \emph{UoI}~\cite{UoI},~\cite{Zheng_UoI}: It measures the estimation error of the underlying monitoring process scaled by an urgency weight which typically increases over time. This metric is similar to Age of Incorrect Information (AoII).
    \item \emph{VoI}~\cite{VoI_with_deterministic_deadline}: The sum of the current values of the packets received by time $t$, each decaying towards a deadline. A packet that is not received has value $0$.
  \item \emph{VIP} (Eq.~\eqref{eq:vip}): $V_i=T-\sum_{j\le i}\tau_j$ if IoTD $z_i$ is collected within the deadline, and $V_i=-\tau_i$ otherwise.
\end{itemize}
\begin{center}
\footnotesize
\setlength{\tabcolsep}{3pt}
\begin{tabular}{|l|c|c|p{1.5cm}|c|}
\hline
\textbf{Metric} & \textbf{Deadline} & \textbf{Failed attempt} & \textbf{Random availability} & \textbf{Shared deadline} \\
\hline
AoI~\cite{Kaul_AoI} & \texttimes & not defined & \texttimes & \texttimes \\
UoI~\cite{Zheng_UoI} & \texttimes & not defined & \texttimes & \texttimes \\
VoI~\cite{VoI_with_deterministic_deadline} & \checkmark & value $0$ & \texttimes & \texttimes \\
VIP & \checkmark & penalty $-\tau_i$ & \checkmark & \checkmark \\
\hline
\end{tabular}
\end{center}
With duty-cycled IoTDs, a probe can fail. AoI and UoI do not account for failed probes, and VoI assigns them no cost, so the cost of probing an unavailable IoTD is not captured. VIP charges this cost in time: the time spent at an IoTD lowers the value of every later IoTD, and an IoTD that is visited but not collected incurs the penalty $-\tau_i$.
}

\section{Standard SA and CMA-ES}\label{app:sa_cmaes}
{
\textit{Simulated annealing.} A better candidate is always 
accepted, and to escape local minima, a worse candidate is accepted with 
probability $\exp\left(\frac{\delta(\mathbf{\hat{z}},\mathbf{z)}}{\Upsilon}\right)$, which 
decreases exponentially as $\Upsilon$ reduces, gradually 
limiting the acceptance of inferior solutions as the 
search progresses.  The function $\delta(\mathbf{\hat{z}},\mathbf{z)}$ represents the difference in scores, as shown in \textit{Line 10 of Algorithm~\ref{alg:SA}}.
\par The temperature $\Upsilon$ controls the degree of 
randomness in the search. At high temperatures, the 
algorithm freely explores the tour space by accepting 
worse solutions, while at low temperatures it 
increasingly focuses on improving solutions. The geometric cooling $\Upsilon\leftarrow\alpha\Upsilon$ gradually shifts the search from exploration to exploitation.

\textit{CMA-ES.} CMA-ES maintains a multivariate Gaussian distribution $\mathcal{N}(\mathbf{m}, 
\sigma^2\mathbf{C})$ over $\mathbb{R}^{2|\mathbf{z}|}$, where $\mathbf{C}$ 
learns the shape of the objective from function evaluations alone, 
without requiring gradient information. This is particularly powerful as $\mathbf{C}$ naturally adapts to the anisotropic objective landscape arising from the geometric interplay between hover position, packet success probability, and 
flight costs to neighboring IoTDs.
\par As shown in Algorithm~\ref{alg:CMAES}, CMA-ES maintains a multivariate Gaussian distribution $\mathcal{N}(\mathbf{m}, \sigma^2\mathbf{C})$ over the search space, where $\mathbf{m}$ is the current best estimate of the solution, $\mathbf{C}$ captures the shape of the objective landscape, and $\sigma$ controls the overall scale of search. At each iteration, the evolution strategy (ES) samples a population of $\Gamma$ candidate solutions from this distribution and evaluates each via a fitness function, selecting the top $\mu = \Gamma/2$ candidates. The covariance matrix adaptation (CMA) then updates $\mathbf{m}$ toward the weighted average of selected candidates, $\mathbf{C}$ to capture the shape and orientation of the selected region, and $\sigma$ via cumulative path length control to regulate the search scale, preventing premature convergence and overshooting. This process repeats until convergence.
}

\onecolumn
\section{Summary of Notations}\label{app:notations}
\begingroup
\footnotesize
\renewcommand{\arraystretch}{1.15}
\begin{longtable}{|c|p{0.68\linewidth}|}
\caption{Summary of notations.}\label{tab:notations}\\
\hline
\textbf{Symbol} & \textbf{Description} \\
\hline
\endfirsthead
\hline
\textbf{Symbol} & \textbf{Description} \\
\hline
\endhead
\hline
\endfoot
$a_k$ & Data-availability probability of IoTD in region $k$ \\
$\mathbf{a}$ & Availability probability vector $(a_i)_{i=1}^{|\mathbf{z}|}$ \\
$A_k$ & Bernoulli random variable representing IoTD $k$'s availability state, $A_k\!\sim\!\text{Bernoulli}(a_k)$ \\
$b_{\max}$ & UAV battery budget \\
$b_f$ & Battery consumed during flight \\
$b$ & Probing energy budget, $b = b_{\max} - b_f$ \\
$B_i$ & Node-wise energy (flight + hover) at IoTD $i$ \\
$B(n_i)$ & Hovering energy as a function of probe count $n_i$ \\
$\mathrm{BW}$ & Channel bandwidth \\
$\mathbf{c}_i$ & Physical location of IoTD $z_i$ (in $\mathbb{R}^2$) \\
$\mathbf{C}$ & CMA-ES covariance matrix \\
$d_i$ & Distance $\lVert \mathbf{l}_i - \mathbf{c}_i\rVert$ from hover to IoTD center \\
$d_{\max}$ & Maximum LoRa coverage range \\
$D(\lambda)$ & Lagrangian dual function \\
$\mathbf{e}_1$ & Unit column vector (first entry $1$, rest $0$) \\
$E$ & Number of evaluations per CMA-ES run \\
$\mathbb{E}[\cdot]$ & Expectation operator \\
$f$ & Flight power \\
$f_c$ & Carrier frequency \\
$f_\psi$ & Probability density function of the SNR $\psi$ \\
$F(\mathbf{l};\mathbf{z})$ & Optimal reward for tour $\mathbf{z}$, hover $\mathbf{l}$ \\
$g,\ g_{\max}$ & CMA-ES generation index / max generations \\
$G(d_i)$ & Distance-dependent channel attenuation / gain term \\
$h$ & Hovering power \\
$H$ & Home location \\
$\mathbf{I}$ & Identity matrix \\
$\mathbf{1}$ & Column vector of ones \\
$\mathbb{I}\{\cdot\}$ & Indicator function \\
$\mathit{itr}$ & SA iteration counter \\
$J^*_t(x_t)$ & Optimal value function at slot $t$ \\
$k$ & Sensing-region / IoTD index \\
$K_{\max}$ & SA iteration budget \\
$\mathbf{l}_i$ & Hover (probing) location at IoTD $z_i$ (in $\mathbb{R}^2$) \\
$\mathbf{l}$ & Hover-position vector $(\mathbf{l}_i)_{i=1}^{|\mathbf{z}|}$ \\
$L_{\mathrm{ex}}$ & Excess path loss \\
$L(\cdot)$ & Value function / Lagrangian (per $\lambda$ argument) \\
$L_{\mathrm{acc}}$ & Lagrangian-value accumulator\\
$\mathbf{m}$ & CMA-ES distribution mean \\
$n_i$ & Number of probes at IoTD $i$ \\
$\mathbf{n}$ & Probe vector $(n_i)_{i=1}^{|\mathbf{z}|}$ \\
$n^\pi_t$ & Probe action at slot $t$ under $\pi$, $\in\{0,1\}$ \\
$\bar{n}$ & Last slot yielding a positive reward \\
$N$ & Number of sensing regions \\
$[N]$ & Index set $\{1,\dots,N\}$ \\
$\mathcal{N}(\mathbf{m},\sigma^2\mathbf{C})$ & CMA-ES multivariate Gaussian sampling distribution \\
$N_0$ & Noise power density \\
$N_{\mathrm{MC}}$ & Number of Monte Carlo rollouts \\
$\mathrm{NF}$ & Noise figure \\
$o_t,\ \mathcal{O}$ & Observation at slot $t$ / observation space \\
$p_k$ & Bernoulli delivery probability of IoTD $k$ \\
$p_t$ & Transmit power of IoTD \\
$P_0,\ P_i$ & Blade profile / induced power constants \\
$\mathcal{P}$ & POMDP transition model \\
$\mathbb{P}(\cdot)$ & Probability measure \\
$\mathbf{Q}$ & Transient matrix of the absorbing chain \\
$R_t$ & Single-stage reward at slot $t$ (random variable) \\
$\mathbf{R}$ & Vector of expected stage rewards \\
$R(\mathbf{z})$ & Optimal reward for tour $\mathbf{z}$ \\
$R_{\mathrm{Exact}},\ R_{\mathrm{Fwd}}$ & Exact vs.\ forward-decoupling reward \\
$\mathcal{R}$ & POMDP reward model \\
$s^{l_i}_i$ & Packet success probability at hover $\mathbf{l}_i$ \\
$t$ & Probe-slot (decision-epoch) index \\
$t_i$ & Hover time at IoTD $i$ \\
$t_{i-1,i}$ & Flight time from IoTD $i{-}1$ to $i$ \\
$t^l_i$ & Per-probe slot duration at IoTD $i$ \\
$T$ & Freshness deadline \\
$u_t,\ \mathcal{U}$ & Action at slot $t$ / action space \\
$v$ & UAV velocity \\
$v_0,\ v_{\mathrm{tip}}$ & Mean rotor induced velocity / blade tip speed \\
$V_i,\ V(n_i)$ & VIP at IoTD $i$ (function of probe count), random variable \\
$V_T(t)$ & Time-varying value of data \\
$W$ & Aircraft weight \\
$x_t,\ \mathcal{X}$ & Belief state at slot $t$ / belief space $[0,1]$ \\
$z_i,\ \mathbf{z},\ \mathcal{Z}$ & Visited IoTD / tour / feasible tour set \\
$\alpha$ & SA geometric cooling rate \\
$\beta$ & Induced power correction factor \\
$\gamma$ & Path-loss exponent \\
$\Gamma$ &  CMA-ES population size \\
$\vartheta$ & Profile drag coefficient \\
$\delta(\hat{\mathbf{z}},\mathbf{z})$ & SA score difference \\
$\Delta_i$ & residual freshness at IoTD $i$ \\
$\epsilon$ & Bisection convergence tolerance \\
$\zeta$ & Rotor radius \\
$\eta$ & Effective receiver noise power \\
$\theta_{SF}$ & SNR threshold of the selected SF \\
$\lambda$ & Lagrange multiplier \\
$\lambda_{\mathrm{low/high/mid}},\ \lambda^*$ & Bisection bounds / midpoint / optimum \\
$\Lambda$ & Rotor solidity \\
$\mu$ & Number of selected CMA-ES candidates, $\lfloor \Gamma/2\rfloor$ \\
$\xi$ & Fuselage drag ratio \\
$\pi,\ \Pi,\ \tilde{\Pi}$ & Policy / policy space / reduced policy space \\
$\psi$ & Received SNR of a packet \\
$\rho$ & Air density \\
$\sigma$ & CMA-ES step size \\
$\tau_i$ & Total flight $+$ hover time, $\tau_i = t_{i-1,i}+t_i$ \\
$\Upsilon$ & SA temperature ($\Upsilon_0$, $\Upsilon_{\min}$) \\
$\phi$ & Absorbing (terminal) state of the node-level POMDP \\
$\Omega$ & Blade angular velocity \\
\end{longtable}
\endgroup

\ifCLASSOPTIONcaptionsoff
  \newpage
\fi
\end{document}